\documentclass{article}
\usepackage[utf8]{inputenc} % allow utf-8 input
\usepackage[T1]{fontenc}    % use 8-bit T1 fonts
\usepackage{hyperref}       % hyperlinks
\usepackage{url}            % simple URL typesetting
\usepackage{booktabs}       % professional-quality tables
\usepackage{amsfonts}       % blackboard math symbols
\usepackage{nicefrac}       % compact symbols for 1/2, etc.
\usepackage{microtype}      % microtypography
\usepackage[ruled,vlined,linesnumbered]{algorithm2e}
\usepackage[disable]{todonotes}
\usepackage{float}
\usepackage{adjustbox}
\usepackage{array} % for better arrays (eg matrices) in maths
\usepackage{amsmath, amsthm, amssymb, color, verbatim}
\usepackage{hyphenat,epsfig,subfigure,multirow}
\usepackage{hyperref}
\usepackage{epstopdf}
\usepackage[font=small,labelfont=bf]{caption}
\usepackage{bmpsize}
\usepackage[font=small,labelfont=bf]{caption}
\usepackage{url}
\newtheorem{theorem}{Theorem}[section]
\newtheorem{lemma}[theorem]{Lemma}
\newtheorem{proposition}[theorem]{Proposition}

\newtheorem{conj}[theorem]{Conjecture}
\newtheorem{definition}{Definition}[section]

\renewcommand{\qed}{\nobreak \ifvmode \relax \else
	\ifdim\lastskip<1.5em \hskip-\lastskip
	\hskip1.5em plus0em minus0.5em \fi \nobreak
	\vrule height0.75em width0.5em depth0.25em\fi}

\title{Massively Parallel Algorithms and Hardness for Single-Linkage Clustering Under $\ell_p$-Distances}
\author{Grigory Yaroslavtsev \thanks{Indiana University, Bloomington, \texttt{grigory@grigory.us}} \and Adithya Vadapalli \thanks{Indiana University, Bloomington, \texttt{avadapal@indiana.edu}}}

\begin{document}

\maketitle

\begin{abstract}
We present massively parallel (MPC) algorithms and hardness of approximation results for computing  Single-Linkage Clustering of $n$ input $d$-dimensional vectors under Hamming, $\ell_1, \ell_2$ and $\ell_\infty$ distances. All our algorithms run in $O(\log n)$ rounds of MPC for any fixed $d$ and achieve $(1+\epsilon)$-approximation for all distances (except Hamming for which we show an exact algorithm).
We also show constant-factor inapproximability results for $o(\log n)$-round algorithms under standard MPC hardness assumptions (for sufficiently large dimension depending on the distance used). Efficiency of implementation of our algorithms in Apache Spark is demonstrated through experiments on a variety of datasets exhibiting speedups of several orders of magnitude. 
	
\end{abstract}

\section{Introduction}
\subsection{Single-linkage clustering}

Single-Linkage Clustering is one of the oldest methods for clustering multi-dimensional vectors based on the nearest-neighbor rule and has been studied since 1951, see e.g. ~\cite{Z71}. 
It can be used for unsupervised learning and  is one of the cornerstone techniques in data mining (see e.g. a classic text on information retrieval for NLP by Manning, Raghavan and Sch\"utze~\cite{MRS08}).
Applications of Single-Linkage Clustering include reconstruction of semantic relationships from word embeddings such as Word2Vec~\cite{ME16}, phylogenetic tree reconstruction~\cite{GR69}, etc.

We consider the problem of constructing a Single-Linkage Clustering for large-scale data.
Formally, given a dataset consisting of vectors $v_1, \dots, v_n \in \mathbb R^d$ the goal is to construct a partition of the vectors into clusters $C_1, \dots, C_k$ such that the smallest distance between two vectors in different clusters is maximized. Formally, for $i\neq j$ let the single linkage distance $d_p(C_i, C_j) = \min_{v_a \in C_i, v_b \in C_j} \|v_a - v_b\|_p$ where $\|x\|_p=(\sum_i |x_i|^p)^{1/p}$. Then in the $k$-Single-Linkage Clustering ($k$-SLC) problem under $\ell_p$-distance  we aim to find a partition into $k$ clusters that maximizes $\min_{i \neq j} d_p(C_i, C_j)$. 
It is well-known that $k$-SLC can be constructed from the Minimum Spanning Tree (MST) of the underlying metric by taking as clusters connected components resulting from removal of $k-1$ longest MST edges.

Note that with this approach once the MST is constructed it can be used to compute $k$-SLC for any value of $k$. Furthermore, it induces a hierarchical clustering structure that is often desirable in practice.
According to ~\cite{MRS08} the main impediment to this approach in practice that motivates the use of various heuristics is that for large-scale high-dimensional data no practically feasible techniques are currently known for constructing an exact MST.

\subsection{Massively parallel computation}
We present analysis of performance of our algorithms in the Massively Parallel Computation model (MPC) which is the most commonly used theoretical model of computation on synchronous large-scale data processing platforms such as MapReduce and Spark. As we demonstrate through experiments in Spark this model accurately reflects performance of our algorithms on real data.
MPC model has attracted a lot of interest recently. It has emerged through a sequence of papers~\cite{FMSSS08, KSV10, GSZ11, BKS13, ANOY14a} and has been analyzed extensively~\cite{FKLRT15, RVW16}. While several variations of this basic model exist here we follow the strictest known version of the model used in~\cite{ANOY14a} and hence our algorithmic results hold in other versions as well.

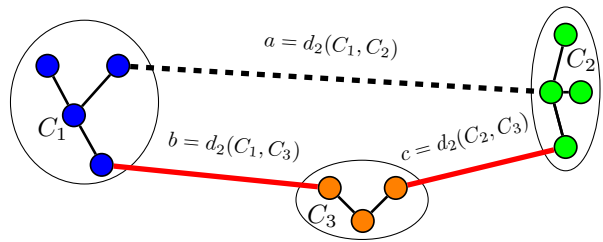
\begin{figure}[ht]
	\begin{center}
		
		\usetikzlibrary{shapes.geometric}
		\trimbox{0cm 5pt 0cm 5pt}{ 
			\begin{tikzpicture}
			[every node/.style={inner sep=0pt}]
			\draw (345.5pt, -131.25pt) ellipse (13pt and 31pt);
			\draw (163.5pt, -136.25pt) ellipse (28pt and 31pt);
			\draw (268.5pt, -173.25pt) ellipse (25pt and 15pt);

			[every node/.style={inner sep=0pt}]
			\node (7) [circle, minimum size=8.5pt, fill=green, line width=0.625pt, draw=black] at (345.5pt, -153.25pt)  {};
			\node (5) [circle, minimum size=8.5pt, fill=green, line width=0.625pt, draw=black] at (340.0pt, -132.5pt)  {};
			\node (6) [circle, minimum size=8.5pt, fill=green, line width=0.625pt, draw=black] at (345.5pt, -110.75pt)  {};
			\node (4) [circle, minimum size=8.5pt, fill=green, line width=0.625pt, draw=black] at (351.25pt, -132.5pt)  {};
			\node (3) [circle, minimum size=8.5pt, fill=blue, line width=0.625pt, draw=black] at (170.0pt, -160.0pt)  {};
			\node (2) [circle, minimum size=8.5pt, fill=blue, line width=0.625pt, draw=black] at (176.25pt, -122.5pt)  {};
			\node (1) [circle, minimum size=8.5pt, fill=blue, line width=0.625pt, draw=black] at (159.5pt, -141.25pt)  {};
			\node (11) [circle, minimum size=8.5pt, fill=blue, line width=0.625pt, draw=black] at (149.5pt, -122.5pt)  {};
			\node (8) [circle, minimum size=8.5pt, fill=orange, line width=0.625pt, draw=black] at (256.25pt, -168.75pt)  {};
			\node (9) [circle, minimum size=8.5pt, fill=orange, line width=0.625pt, draw=black] at (281.25pt, -168.75pt)  {};
			\node (10) [circle, minimum size=8.5pt, fill=orange, line width=0.625pt, draw=black] at (268.75pt, -181.25pt)  {};
			\draw [line width=1, color=black] (3) to  (1);
			\draw [line width=1, color=black] (11) to  (1);
			\draw [line width=1, color=black] (1) to  (2);
			\draw [line width=1, color=black] (10) to  (8);
			\draw [line width=1, color=black] (10) to  (9);
			\draw [line width=1, color=black] (5) to  (7);
			\draw [line width=1, color=black] (5) to  (7);
			\draw [line width=1, color=black] (5) to  (6);
			\draw [line width=1, color=black] (5) to  (4);
			\draw [line width=2, color=red] (3) to  (8);
			\draw [line width=2, color=red] (7) to  (9);
			\draw [line width=2, color=black, style=dashed] (2) to  (5);
			\node at (151.5pt, -146.25pt) {\textcolor{black}{$C_1$}};
			\node at (351.5pt, -120.25pt) {\textcolor{black}{$C_2$}};
			\node at (253.5pt, -179.25pt) {\textcolor{black}{$C_3$}};
			\node at (256.875pt, -114.75pt) [scale = 0.8, rotate=356]{\textcolor{black}{$a = d_2(C_1, C_2)$}};
			\node at (220.25pt, -153.125pt) [scale = 0.8, rotate=355] {\textcolor{black}{$b = d_2(C_1, C_3)$}};
			\node at (307.375pt, -150.0pt) [scale = 0.8, rotate=15] {\textcolor{black}{$c = d_2(C_2, C_3)$}};
			\end{tikzpicture}
		}		
		\caption{ 3-SLC objective is $\min(a,b,c)$,  MST shown in solid.}              \label{fig:slc-ex}
	\end{center}
\end{figure}

\begin{figure}
	\begin{center}

	\usetikzlibrary{shapes.geometric}
	\usetikzlibrary{decorations.pathreplacing}
	
	\begin{tikzpicture}

	\usepgflibrary{arrows}
	\pgfmathsetmacro{\cubex}{1.5}
	\pgfmathsetmacro{\cubey}{0.25}
	\pgfmathsetmacro{\cubez}{0.8}
	\draw (0,0,0) -- ++(-\cubex,0,0) -- ++(0,-\cubey,0) -- ++(\cubex,0,0) -- cycle;
	\draw (0,0,0) -- ++(0,0,-\cubez) -- ++(0,-\cubey,0) -- ++(0,0,\cubez) -- cycle;
	\draw (0,0,0) -- ++(-\cubex,0,0) -- ++(0,0,-\cubez) -- ++(\cubex,0,0) -- cycle;
	
	\draw (2.5,0,0) -- ++(-\cubex,0,0) -- ++(0,-\cubey,0) -- ++(\cubex,0,0) -- cycle;
	\draw (2.5,0,0) -- ++(0,0,-\cubez) -- ++(0,-\cubey,0) -- ++(0,0,\cubez) -- cycle;
	\draw (2.5,0,0) -- ++(-\cubex,0,0) -- ++(0,0,-\cubez) -- ++(\cubex,0,0) -- cycle;
	
	\draw (5.5,0,0) -- ++(-\cubex,0,0) -- ++(0,-\cubey,0) -- ++(\cubex,0,0) -- cycle;
	\draw (5.5,0,0) -- ++(0,0,-\cubez) -- ++(0,-\cubey,0) -- ++(0,0,\cubez) -- cycle;
	\draw (5.5,0,0) -- ++(-\cubex,0,0) -- ++(0,0,-\cubez) -- ++(\cubex,0,0) -- cycle;

	\draw (0,1,0) -- ++(-\cubex,0,0) -- ++(0,-\cubey,0) -- ++(\cubex,0,0) -- cycle;
	\draw (0,1,0) -- ++(0,0,-\cubez) -- ++(0,-\cubey,0) -- ++(0,0,\cubez) -- cycle;
	\draw (0,1,0) -- ++(-\cubex,0,0) -- ++(0,0,-\cubez) -- ++(\cubex,0,0) -- cycle;
	
	\draw (2.5,1,0) -- ++(-\cubex,0,0) -- ++(0,-\cubey,0) -- ++(\cubex,0,0) -- cycle;
	\draw (2.5,1,0) -- ++(0,0,-\cubez) -- ++(0,-\cubey,0) -- ++(0,0,\cubez) -- cycle;
	\draw (2.5,1,0) -- ++(-\cubex,0,0) -- ++(0,0,-\cubez) -- ++(\cubex,0,0) -- cycle;
	
	\draw (5.5,1,0) -- ++(-\cubex,0,0) -- ++(0,-\cubey,0) -- ++(\cubex,0,0) -- cycle;
	\draw (5.5,1,0) -- ++(0,0,-\cubez) -- ++(0,-\cubey,0) -- ++(0,0,\cubez) -- cycle;
	\draw (5.5,1,0) -- ++(-\cubex,0,0) -- ++(0,0,-\cubez) -- ++(\cubex,0,0) -- cycle;

	\draw (0,2,0) -- ++(-\cubex,0,0) -- ++(0,-\cubey,0) -- ++(\cubex,0,0) -- cycle;
	\draw (0,2,0) -- ++(0,0,-\cubez) -- ++(0,-\cubey,0) -- ++(0,0,\cubez) -- cycle;
	\draw (0,2,0) -- ++(-\cubex,0,0) -- ++(0,0,-\cubez) -- ++(\cubex,0,0)-- cycle  ;
	
	\draw (2.5,2,0) -- ++(-\cubex,0,0) -- ++(0,-\cubey,0) -- ++(\cubex,0,0) -- cycle;
	\draw (2.5,2,0) -- ++(0,0,-\cubez) -- ++(0,-\cubey,0) -- ++(0,0,\cubez) -- cycle;
	\draw (2.5,2,0) -- ++(-\cubex,0,0) -- ++(0,0,-\cubez) -- ++(\cubex,0,0) -- cycle;
	
	\draw (5.5,2,0) -- ++(-\cubex,0,0) -- ++(0,-\cubey,0) -- ++(\cubex,0,0) -- cycle;
	\draw (5.5,2,0) -- ++(0,0,-\cubez) -- ++(0,-\cubey,0) -- ++(0,0,\cubez) -- cycle;
	\draw (5.5,2,0) -- ++(-\cubex,0,0) -- ++(0,0,-\cubez) -- ++(\cubex,0,0) -- cycle;
	
	\draw [-stealth] (0,1.75,0) -- (1,0.75,0) node{};
	\draw [-stealth] (0,0.75,0) -- (1,-0.25,0) node{};
	\draw [-stealth] (0,-0.25,0) -- (1,1.75,0) node{};
	\draw [-stealth] (0, 1.75, 0) -- (1, 1.75, 0) node{}; 
	
	\draw [-stealth] (2.5,-0.25,0) -- (4,1.75,0) node{};
	\draw [-stealth] (2.5,0.75,0) -- (4,-0.25,0) node{};
	\draw [-stealth] (2.5,1.75,0) -- (4,0.75,0) node{};
	\draw [red, very thick] (0.25,1.35) arc (-40:10:0.65) ;
	\draw [red, very thick] (0.8, 1.9) arc (-25:40:-0.65) ;

	\draw [decorate,decoration={brace,amplitude=10pt,mirror,raise=4pt},yshift=0pt, thick]
	(-1.4,2.5) -- (-1.4, -0.5) node [black,midway,xshift=-0.75cm] {\rotatebox{90}{$m$ machines}};
	
	\draw [decorate,decoration={brace,amplitude=10pt,mirror,raise=4pt},yshift=0pt, thick]
	(-1.5,-0.5) -- (5.8,-0.5) node [black,midway,yshift=-1.0cm] {$R$ Rounds};
	
	\draw [white,decorate,decoration={brace,amplitude=10pt,mirror,raise=4pt},yshift=0pt, thick]
	(2,2.5) -- (-0.5, 2.5) node [black,midway,yshift= 0.2cm] {\textcolor{red}{$\boldmath{\le s}$} bits sent/received};
	\end{tikzpicture}	
	
	\caption{MPC model of computation}\label{fig:mpc}
\end{center}
\end{figure}
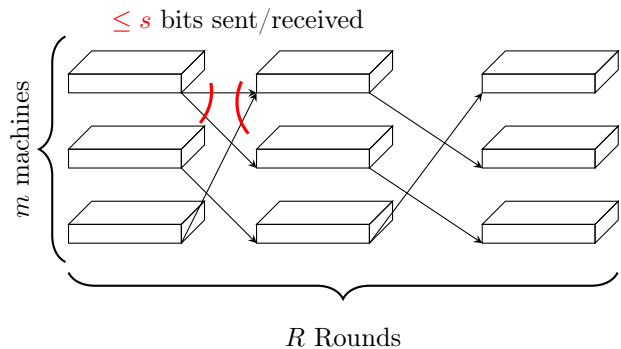

\vspace{20pt}

\vspace{10pt}

In the MPC model we are given access to $m$ identical processors with local RAM space $s$ on each. For an input of size $n$ the total space available to all processors is $m \cdot s = \tilde O(n)$.  The computation is performed in synchronous rounds. In each round each machine: 1) performs a local computation on its data (under its local space restriction of $s$), 2) sends and receives messages of total length at most $s$ to other machines which are received before the next round begins (see Figure~\ref{fig:mpc}).
Note that restriction of $s$ on the total length of received messages follows from the local space constraint assuming there is no computation performed on the fly on incoming data. 
Furthermore, we assume that the most time/space-efficient known algorithm for local subproblems (in our case almost linear-time and space) is used on each machine during the round.

In this setup the key complexity measure of performance in such computation is the number of rounds it takes to complete it as other characteristics such as time and communication depend directly on it. The parameter $s$ is set to $n^{\alpha}$ for some fixed constant $\alpha < 1$, see~\cite{KSV10, ANOY14a} for more details. In this setting of parameters sorting can be done in $O(1)$ rounds~\cite{GSZ11} while sparse graph connectivity takes $O(\log n)$\cite{RMCS13, KLMRV14} which is conjectured to be optimal~\cite{KSV10, BKS13,RMCS13, RVW16}.  It also appears to be folklore that an $O(\log n)$-round algorithm for MST in sparse graphs can be obtained via a simulation of Boruvka's algorithm in MPC. We use these facts extensively in this paper.

\subsection{Our results and previous work}
While scalable algorithms with provable guarantees for other popular clustering methods such as $k$-means and  $k$-median are known ~\cite{BMVKV12, BEL13} we are not aware of any such algorithms for Single-Linkage Clustering (with the exception of recent work of~\cite{DBBMHLK17} who consider a more general graph metric setting and hence get results which are inherently different from our work).
Also despite the fact that scalable heuristics exist for $k$-SLC and MST computation for vector data, e.g.~\cite{J15}, the only MPC algorithm with provable guarantees in this area that we are aware of is~\cite{ANOY14a}\footnote{For general graph metrics an MST algorithm in MPC is given in~\cite{KSV10}. In our case using this algorithm directly would imply a quadratic increase in space since our graph is implicitly given by $n^2$ distances between the vectors and hence constructing the graph explicitly is infeasible under the overall space restriction.}. For other recent work on geometric data structures and algorithms in the MPC model see~\cite{AFMN16,NFMA16} and results on distributed constructions of coresets~\cite{AHV05,IMMM14,BBLM14}.

In~\cite{ANOY14a} it is shown that a $(1+\epsilon)$-approximate MST under $\ell_2^d$ can be constructed in $O(1)$ rounds of MPC for constant dimension.
However, while the overall cost of the MST is a good approximation to the optimum the length of any given edge can be arbitrarily distorted. This makes it impossible to directly use the algorithm of~\cite{ANOY14a} for the Single-Linkage Clustering problem. For example, consider an input corresponding to a set of points on the line shown in Figure~\ref{fig:bad-ex} and $k = 2$. In this case a $(1 + \epsilon)$-approximate MST would not necessarily lead to a $(1 + \epsilon)$-approximate clustering as any such clustering would have to have clusters $\{1, \dots, n - 1\}$ and $\{n\}$ which are at distance $100$ from each other.
Moreover, the algorithm of~\cite{ANOY14a} will indeed introduce edges of length $\Omega(\epsilon n)$ into its approximate MST between the first $n - 1$ points if run on this example.  Hence for the MST constructed using~\cite{ANOY14a} the basic approach of removing the longest edge to obtain a $2$-SLC will result in two clusters which are at distance $1$ with a very large probability.

\begin{figure}
	\begin{center} 
	\usetikzlibrary{shapes.geometric}
	\begin{tikzpicture}
	[every node/.style={inner sep=0pt}]
	\node (1) [circle, minimum size=7.5pt, fill=black, line width=0.625pt, draw=black] at (31.25pt, -56.25pt)  {};
	\node (2) [circle, minimum size=7.5pt, fill=black, line width=0.625pt, draw=black] at (50.0pt, -56.25pt)  {};
	\node (3) [circle, minimum size=7.5pt, fill=black, line width=0.625pt, draw=black] at (68.75pt, -56.25pt)  {};
	\node (4) [circle, minimum size=7.5pt, fill=black, line width=0.625pt, draw=black] at (87.5pt, -56.25pt)  {};
	\node (5) [circle, minimum size=7.5pt, fill=black, line width=0.625pt, draw=black] at (106.25pt, -56.25pt)  {};
	\node (6) [circle, minimum size=7.5pt, fill=black, line width=0.625pt, draw=black] at (125.0pt, -56.25pt)  {};
	\node (7) [circle, minimum size=7.5pt, fill=black, line width=0.625pt, draw=black] at (206.25pt, -56.25pt)  {};
	\draw [line width=0.625, color=black] (1) to  (2);
	\draw [line width=0.625, color=black] (2) to  (3);
	\draw [line width=0.625, color=black] (3) to  (4);
	\draw [line width=0.625, color=black] (4) to  (5);
	\draw [line width=0.625, color=black] (5) to  (6);
	\draw [line width=0.625, color=black] (6) to  (7);
	\node at (31.25pt, -70.0pt) {\textcolor{black}{1}};
	\node at (50.0pt, -70.0pt) {\textcolor{black}{2}};
	\node at (125.0pt, -70.0pt) {\textcolor{black}{$n$ - 1}};
	\node at (206.25pt, -70.0pt) {\textcolor{black}{$n$}};
	\end{tikzpicture}
	
	\caption{$\forall i \le n-1 \colon \|v_{i-1} - v_{i}\| = 1 $, $\|v_{n - 1} - v_n\|= 100$.}              \label{fig:bad-ex}
\end{center}
\end{figure}
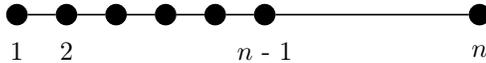

In this paper we show how to overcome this difficulty and give a different algorithm inspired by~\cite{ANOY14a} that allows to compute an approximate Single-Linkage Clustering.
While in~\cite{ANOY14a} only $\ell_2^d$ metric is considered here we further extend this framework so that it also applies to $\ell_1^d$ and $\ell_\infty^d$ with similar performance guarantees.
Perhaps most interestingly, while an arbitrarily good MST approximation can be computed in $O(1)$ rounds of MPC (for fixed dimension) our algorithms for $k$-SLC run in $O(\log n)$ rounds. 
As it turns out, such an increase is likely to be necessary.
We justify it through a number of (conditional) hardness results. Our results show that even for $k=2$ assuming two most popular conjectures in the MPC literature regarding complexity of sparse connectivity no $o(\log n)$-round algorithm can compute $k$-SLC for sufficiently large dimension of the data with better than some fixed constant-factor approximation that depends on the distance metric used. See Table~\ref{tbl:results} for a summary of these results\footnote{While our algorithms work in the weakest version of MPC model, our hardness results also hold in stronger versions for which hardness of sparse connectivity is conjectured, see~\cite{RVW16} for further details. Furthermore, in hardness results for $\ell_0$ and $\ell_1$ that require dimension $d = \Omega(n)$ the result holds for $O(1)$-sparse vectors, i.e. the overall input size is still $O(n)$ words.}.

In order to complete the picture of approximability of $k$-SLC under the most frequently used $\ell_p$-distances we also give algorithms and hardness results under Hamming distance (commonly referred to as $\ell_0$). In contrast to other distances studied in this paper we are able to completely resolve approximability of the $k$-SLC problem for constant-dimensional data in this case.
As we show, there exists an exact algorithm for $d = O(1)$ that runs in $O(\log n)$ rounds of MPC while under Conjecture~\ref{conj:connectivity} no algorithm running in $o(\log n)$ rounds can obtain better than $2$-approximation even for $d = 2$. See Table~\ref{tbl:results} for details.

\begin{table}[t]
	\caption{Approximation and hardness of $k$-SLC in MPC under $\ell_p$-distances}
	\label{tbl:results}
	\centering
	\begin{tabular}{lll}
		\toprule
		%		\multicolumn{2}{c}{Part}                   \\
		\cmidrule{1-3}
		& Approximation in $O(\log n)$ rounds   &  Hardness of approximation in $o(\log n)$ rounds \\
		\midrule
		$\ell_0$ & Exact for $d = O(1)$, Thm.~\ref{thm:hamming-k-slc}  &   $2$ for $d = 2$  under Conj.~\ref{conj:connectivity}, Thm.~\ref{thm:hamming-k-slc-hardness}\\
		&&$3$ for $d = \Omega(n)$  under Conj.~\ref{conj:cycles}, Thm.~\ref{thm:l1l2-k-slc-hardness}\\
		\midrule
		$\ell_1$     & $(1 + \epsilon)$  for $d = O(1)$, Thm.~\ref{thm:k-slc-main} &  $3$ for $d = \Omega(n)$ under Conj.~\ref{conj:cycles}, Thm.~\ref{thm:l1l2-k-slc-hardness}    \\
		&& $2$ for $d = \Omega(n)$ under Conj.~\ref{conj:connectivity}, Thm.~\ref{thm:l1l2-k-slc-hardness} \\
		\midrule
		$\ell_2$     &  $(1 + \epsilon)$ for $d = O(1)$, Thm.~\ref{thm:k-slc-main} &    $1.84 - \epsilon$ for $d = \Omega(\frac{\log n}{\epsilon^2})$ under Conj.~\ref{conj:cycles}, Thm.~\ref{thm:l1l2-k-slc-hardness} \\
		&&$1.41 - \epsilon$ for $d = \Omega(\frac{\log n}{\epsilon^2})$ under Conj.~\ref{conj:connectivity}, Thm.~\ref{thm:l1l2-k-slc-hardness} \\
		\midrule
		$\ell_\infty$     &  $(1 + \epsilon)$ for $d = O(1)$, Thm.~\ref{thm:k-slc-main} &    $2$ for $d = \Omega(\log n)$ under Conj.~\ref{conj:connectivity}, Thm. 7.1 ~\cite{ANOY14b} \\
		\bottomrule
	\end{tabular}
\end{table}

\subsection{Our techniques}
Our algorithms under $\ell_1, \ell_2$ and $\ell_\infty$ all share the same high-level structure: we tackle the problem of the input having $O(n^2)$ edges by first constructing a sparsifier that only has $O(n \log n)$ edges and then run an MST algorithm on this sparsifier.  In order to construct a sparsifier we execute the $(1 + \epsilon)$-approximate MST algorithm of~\cite{ANOY14a} $O(\log n)$ times and collect all edges of the MSTs constructed in these executions. We then run an exact $O(\log n)$-round exact MST algorithm on this set of $O(n \log n)$ edges and output clusters resulting from removing $k - 1$ longest edges of the resulting MST.
Note that the executions of~\cite{ANOY14a} can be done in parallel and hence it is the second step that introduces $O(\log n)$ rounds into the overall complexity of the algorithm.  Our algorithms under $\ell_1, \ell_2$ and $\ell_\infty$ are given in Section~\ref{sec:algorithms} with some specific technical details deferred to Section~\ref{sec:implementation}.
Assuming the same high-level structure this approach is unlikely to be improved as there are no known algorithms for solving MST in sparse graphs in $o(\log n)$ rounds.

In fact, we make the above observation formal by giving reductions from two most popular problems conjectured to require $\Omega(\log n)$ rounds in the MPC model: sparse connectivity (Conjecture~\ref{conj:connectivity}) and a stronger ``one cycle vs. two cycles'' problem (Conjecture~\ref{conj:cycles}).
Our reductions follow the same general strategy -- we introduce a vector $v_i \in \mathbb R^n$ for each vertex in the input graph. This vector is initially set to be $e_i$, the $i$-th standard unit vector. Then for each edge $(i,j)$ adjacent to the vertex $i$ we update the coordinate $j$ of the vector by adding a carefully chosen value $\xi$. This ensures that the for pairs of points which are connected by an edge the distance between their correponding vectors is different from the distance between points which are not connected by an edge. The parameter $\xi$ is then chosen to maximize the ratio of between distances in these two cases. Details are given in Section~\ref{sec:hardness}.

Under $\ell_0$ (Hamming distance) we can't construct a $(1 + \epsilon)$-approximate MST using~\cite{ANOY14a}  and hence our algorithms and hardness results are quite different.
Using sorting as a primitive we construct an auxiliary graph and then run an $O(\log n)$-round connectivity algorithm on it $d$ times. This way we obtain an exact MST and hence an exact $k$-SLC for any value of $k$.
Details are given in Section~\ref{sec:exact-hamming-mst}. Our hardness reduction in this case is also quite different as we construct a hard instance by creating a set of points in 2D instead of using high-dimensional vectors. Hence our result rules out an $o(\log n)$-round $2$-approximation even for $d = 2$. See Section~\ref{sec:hamming-hardness} for details.
\vspace{10pt}
\subsection{Experimental results}
We implemented our algorithm (for $\ell_2$ distances) in Java on Apache Spark and empirically evaluated the performance. The largest dataset we used was the US Census data from the UCI ML repository which has been used widely in literature ($n = 2458285$, $d=8$). 
Note that storage of the $n^2$ adjacency would take nearly 24TB of memory and hence building a complete graph locally is infeasible using commodity hardware.  We observed speedups of several orders of magnitude compared to our benchmark sequential Prim's algorithm when using 200 reducers. We remark that the speedup is not just due to the parallelism in our algorithm but also due to the use of approximation which is helpful even if the algorithm is executed locally. See Section~\ref{sec:experiments} for details.

\section{Algorithms}\label{sec:algorithms}

\subsection{Partition-based algorithm}

\begin{theorem}\label{thm:k-slc-main}
	For each of the three metrics $\ell_1^d, \ell_2^d$ and $\ell_\infty^d$ for any constants $0 < \eta \le 3$, $0 < \alpha < 1/2$ and $d = O(1)$ there exists a constant $\kappa > 0$ such that if $(\kappa \eta)^{-2d} \le s^{1 - 2 \alpha}$ then there exists an $O(\log n)$-round MPC algorithm that computes $(1 + \eta)$-approximate $k$-Single-Linkage Clustering of an input set of vectors $v_1, \dots, v_n \in \mathbb R^d$ for all values of $k$ under these metrics. The algorithm is randomized and produces correct result with high probability. Given access to machines with RAM space $s$ it uses $\tilde O(n/s)$ machines and time at most $\tilde O(s)$ per round on each machine.
\end{theorem}

In this section we describe a generic partition-based algorithm, Algorithm~\ref{alg:partition-mst}, that is used to prove the above theorem. We also give analysis of its approximation guarantee. Algorithm~\ref{alg:partition-mst} relies on $(a,b,c)$-distance-preserving partitions and uses Algorithm~\ref{alg:unit-step} which we describe in Section~\ref{sec:implementation}.

\begin{algorithm}
	\SetKwInOut{Input}{input}
	\SetKwInOut{Output}{output}
	\Input{Vectors $v_1, \dots, v_n \in \mathbb R^d$, parameters $\eta, \alpha, p$}
	\Output{Approximate MST for $v_1, \dots, v_n$ under $\ell_p^d$}
	\caption{Generic Partition-based Distributed Single-Linkage Clustering Algorithm}\label{alg:partition-mst}
	$E = \emptyset$ \\
	Set $a = 1/(s^{\alpha/d} -1), b=poly(d), c = s^{\alpha}, L = O(\log_{1/a} n), \epsilon = \min\left(\frac{\eta}{6 c_1 L b}, \frac{\eta}{3 c_2}\right)$\\
	Repeat $O(\log n)$ times sequentially: \\
	\quad\quad\quad Sample partition $P$ with $L$ levels from \\ \quad \quad \quad \quad $(a,b,c)$-distance-preserving  family $w_P$ \label{step:partition}\\
	\quad \quad \quad Execute unit step Algorithm~\ref{alg:unit-step} with $\rho = \ell_p^d$ for \\ 
	\quad \quad \quad \quad each cell in $P$ \label{step:recursion} with parameter $\epsilon$\\
	\quad \quad \quad $E' = $ set of edges output in the previous step \\
	\quad \quad\quad $E = E \cup E'$ \\
	Run Boruvka's MST algorithm on $G = (V,E)$ \label{step:boruvka}\\
\end{algorithm}

Let $M(S,\rho)$ be a metric space and  $w \colon S \times S \rightarrow \mathbb R^+$ be a weight function $w(x,y) = \rho(x,y)$. We think of $w$ as representing weights of edges in a complete graph. Let $MST_i(w)$ denote the weight of the $i$-th Minimum Spanning Tree edge of this graph sorted in non-decreasing order.
Let $w^+ \colon S \times S \rightarrow \mathbb R^+$ be a random family of functions that satisfies that for each $x,y$ it holds that $w(x,y) \le w^+(x,y)$ and $\mathbb E[w^+(x,y)] \le (1 + \gamma) w(x,y)$ for some fixed $\gamma > 0$.
Note that the weights given by this random family to different pairs might be correlated with each other.

\begin{definition}[Crossing edge]
	For a partition $(C_1, \dots, C_t)$ of $S$ we say that a pair of points $(x,y)$ \textit{crosses} this partition if $x \in C_i$ and $y \in C_j$ for $i \neq j$.
\end{definition}
\begin{definition}[Cut-preserving spanning tree]
	We say that $\mathcal T$ is an $\alpha$-cut-preserving spanning tree for $w: S \times S \rightarrow \mathbb R^+$ if for every partition $(C_1,C_2)$ of $S$ there exists an edge in $\mathcal T$ that crosses this partition and is  at most $\alpha$ times longer than the shortest such edge with respect to $w$.
\end{definition}

As we show below Algorithm~\ref{alg:partition-mst} can be seen as performing the following experiment: draw $k$ functions $w_1, \dots, w_k$ i.i.d at random from the family $w^+$.
Compute a $(1 + \delta)$-cut-preserving spanning tree $\mathcal T_i$ for each $w_i$.
Then for each $(x,y) \in S \times S$ define $w'_i(x,y) = w(x,y)$ if $(x,y)$ is in this spanning tree and $w'_i(x,y) = +\infty$ otherwise.
Then for all $(x,y) \in S \times S$ define $\bar w^k(x,y) = \min_{i = 1}^k w'_i(x,y)$.
The final run of Boruvka's MST algorithm is then executed on $\bar w^k$.

Indeed, random family of functions $w^+$ satisfying the properties described above is constructed by Algoirthm~\ref{alg:partition-mst} as follows from a result~\cite{ANOY14a} given. It is important to note that cut-preserving spanning tree computations for random function samples from this family required above can be also performed as guaranteed by the following lemma:

\begin{lemma}[\cite{ANOY14a}, Lemma 3.4 and Lemma 3.13]\label{lem:anoy-metric}
	Given access to an $(a,b,c)$-distance-preserving partition with $L$ levels and approximation $\gamma$ for $M(S,\rho)$  there exists an MPC algorithm that runs in $O(1)$ rounds and constructs a random family of weight functions $w_P$ which satisfies: 
	$$\rho(i,j) \le w_P(i,j) \text { and } \mathbb E[w_P(i,j)] \le \left(1 + c_1 \epsilon L b\right) \rho(i,j).$$
	Furthermore, execution of unit step Algorithm~\ref{alg:unit-step} for all cells in this partition for a random function $w^*$ sampled from $w_P$ produces a $(1+c_2\epsilon)$-cut-preserving spanning tree $\mathcal T$ for $w^*$. 
	
\end{lemma}

Let $w(i,j) = \|v_i - v_j\|_2$, $w^+ = w_P$ and let $\gamma = c_1 \epsilon d$ and $\delta = c_2 \epsilon$.

\begin{lemma}\label{lem:mst-edge-approx}
	Let $n = |S|$. There is a large enough constant $c >0$ such that if $k = c \log n$ then for all $i$ it holds that:
	$$\Pr_{w_1, \dots, w_k}[MST_i(\bar w_k) \le (1 + 2 \gamma)(1 + \delta) MST_i(w)] \ge 1 - 1/poly(n).$$
\end{lemma}
\begin{proof}
	Fix $(x,y) \in S \times S$ and let $\Delta(x,y) = w^+(x,y) - w(x,y)$.
	Because $\Delta(x,y) \ge 0$ and $\mathbb E[\Delta(x,y)] \le \gamma w(x,y)$ with probability at least $1/2$ it holds that $\Delta(x,y) \le 2 \gamma w(x,y)$ by Markov inequality. If $k = c \log n$ then with probability $1 - 1/n^c$ there exists $i$ such that $w_i(x,y) - w(x,y) \le 2 \gamma w(x,y)$. By a union bound over all $n^2$ pairs $(x,y)$ with probability $1 - 1/n^{c-2}$ for each such pair a corresponding index exists. Below we refer to this event as $\mathcal E$ and condition on it.

	\begin{proposition}\label{prop:shortest-edge-approx}
		Let $(C_1, \dots, C_t)$ be an arbitrary partition of $S$.
		Let $(x^*,y^*) \in S \times S$ be the closest w.r.t $w$ pair of points that belong to different parts of this partition.
		Then conditioned on the event $\mathcal E$ there exists a pair of points $(x',y')$ that crosses this partition and:
		$$w(x^*,y^*) \le \bar w^k(x',y') \le (1 + 2\gamma) (1 + \delta) w(x^*, y^*).$$ 
	\end{proposition}
	\begin{proof}
		First, consider the case when $t = 2$ and consider any partition $(C_1, C_2)$ of $S$.
		Let $(x^*,y^*)$ be the shortest edge that crosses this partition, i.e. $(x^*, y^*) := arg\min_{x \in C_1, y \in C_2} w(x,y)$.
		Conditioned on $\mathcal E$ there exists $i$ such that $w_i(x^*,y^*) \le (1 + 2\gamma) w(x^*,y^*)$.
		Furthermore, there exists an edge $(x',y')$ in the $(1 + \delta)$-cut-preserving spanning tree $\mathcal T_i$ constructed for $w_i$ that has length $w'_i(x',y') = w_i(x',y') \le (1 + \delta) w_i(x^*,y^*) \le (1 + 2\gamma)(1 + \delta) w(x^*,y^*)$.
		On the other hand, because $w_i \ge w$ for every pair $(x,y)$ that crosses the partition $(C_1,C_2)$ it holds that $w_i(x,y) \ge w(x^*,y^*)$.
		Combining these two facts we conclude that in $\mathcal T_i$ there exists some edge $(x', y')$ that crosses the cut and satisfies $w(x^*,y^*) \le w'_i(x',y') \le (1 + 2\gamma)(1 + \delta)w(x^*,y^*)$. By definition of $\bar w^k$ the same holds for it as well, i.e. $w(x^*,y^*) \le \bar w^k(x',y') \le (1 + 2\gamma)(1 + \delta)w(x^*,y^*)$.
		
		Now suppose $t > 2$.
		For $i = 1, \dots, t$ define a family of cuts $(S_i, T_i)$ where $S_i = C_i$ and $T_i = \cup_{j \neq i} C_j$. Let $(x^*_i, y^*_i)$ be the shortest pair crossing the cut $(S_i, T_i)$.
		If $(x^*,y^*)$ is the shortest edge that crosses $(C_1, \dots, C_t)$ then  we have $w(x^*,y^*) = \min_i w(x^*_i, y^*_i)$. Let $i^* = arg \min_i w(x^*_i, y^*_i)$.
		Then using the argument above for $t = 2$ there exists $(x',y')$ such that $x' \in S_{i^*}, y \in T_{i^*}$ and:
		\begin{align*}
		w(x^*,y^*) & = w(x^*_{i^*}, y^*_{i^*}) \\
		&\le \bar w^k(x',y') \\
		&\le (1 + 2\gamma)(1 + \delta)w(x^*_{i^*}, y^*_{i^*}) \\
		&= (1 + 2\gamma)(1 + \delta)w(x^*,y^*).\qedhere
		\end{align*}
		
	\end{proof}
	
	Given Proposition~\ref{prop:shortest-edge-approx} the rest of the proof is the same as analysis of approximate Kruskal's algorithm in~\cite{I00}, we give the proof here for completeness.
	Since edges output by Kruskal's algorithm are produced in the order of non-decreasing weight $MST_i$ is the $i$-th edge that is output.
	Consider executions of Kruskal's algorithm on weights $w$ and $\bar w^k$.
	Let the edges output by the former execution be $e_1, \dots, e_{n-1}$ in order.
	Let the edges output by the latter execution be $e'_1, \dots, e'_{n-1}$. 
	
	To prove Lemma~\ref{lem:mst-edge-approx} it suffices to show that conditioned on $\mathcal E$ it holds that $w(e_i) \le w(e'_i) \le (1 + 2\gamma)w(e_i)$ for all $i$. The first inequality here essentially follows from the fact that the weight of the $i$-th MST edge is a monotone function of the weights and $w \le \bar w^k$.
	
	The $i$-th edge in Kruskal's algorithm is constructed by joining two closest clusters among $n -i + 1$ clusters constructed so far. Let these clusters  in the execution of Kruskal's algorithm on $\bar w^k$ be denoted as $C_1, \dots, C_{n- i + 1}$.
	The key observation is that there exists an index $i^* \le i$ such that endpoints of the edge $e_{i^*}$ belong to different parts of the partition $C_1, \dots, C_{n-i + 1}$. Indeed, edges $e_1, \dots, e_i$ form a forest and thus having all such edges be inside $C_1, \dots, C_{n-i + 1}$ would be a contradiction.
	
	Let $(x^*,y^*)$ be the closest w.r.t to $w$ pair of points in different parts of the partition $C_1, \dots, C_{n - i + 1}$.
	By applying Proposition~\ref{prop:shortest-edge-approx} to $e_{i^*}$ there exists a pair of points $(x',y')$ whose endpoints belong to different parts of the partition $C_1, \dots, C_{n - i + 1}$ and $\bar w^k(x',y') \le (1 + 2\gamma) w(x^*, y^*)$.
	Putting everything together we have:
	
	\begin{align*}
	w(e_i') & \le \bar w^k(e'_i) && \text{$w \le \bar w^k$} \\ 
	&\le \bar w^k(x',y') \\ 
	&\le (1 + 2\gamma)(1 + \delta) w(x^*, y^*) && \text{Proposition~\ref{prop:shortest-edge-approx}}\\
	&\le (1 + 2\gamma)(1 + \delta) w(e_{i^*}) &&\text{$e_{i*}$ crosses $(C_1, \dots, C_{n - i + 1})$}\\
	&\le (1 + 2\gamma)(1 + \delta) w(e_i) && \qedhere
	\end{align*}
	
	The second inequality follows because $e'_i$ is shortest edge w.r.t $\bar w^k$ that crosses  $(C_1, \dots, C_{n - i + 1})$. The last inequality follows because $i^* \le i$, edge weights are non-decreasing.

\end{proof}

Putting everything together we obtain analysis of approximation guaranteed by Algorithm~\ref{alg:partition-mst}.

\begin{theorem}\label{thm:mst-approximation}
	For $\eta \le 3$ and $p = 1, 2, \infty$ Algorithm~\ref{alg:partition-mst} constructs a spanning tree for $w(i,j) = \|v_i - v_j\|_p$ for each $t$ its $t$-th longest edge $(x,y)$ has weight $w(x,y) \le (1 + \eta)MST_k(w)$. This guarantee holds with high probability over the randomness used in Algorithm~\ref{alg:partition-mst}.
\end{theorem}

\begin{proof}
	Note that taking $w^+ = w_P$ for $w(i,j) = \|v_i - v_j\|_p$ where $p = 1, 2, \infty$ satisfies conditions of Lemma~\ref{lem:mst-edge-approx} by Lemma~\ref{lem:anoy-metric}. Hence our algorithm constructs a function $\bar w_k$ with properties required for Lemma~\ref{lem:mst-edge-approx}.
	Since $c_1 \epsilon L b \le \eta/6$ and $c_2 \epsilon \le \eta/3$ we can set $\delta = \eta/6$ and $\gamma = \eta/3$ in Lemma~\ref{lem:mst-edge-approx} and hence for $\eta \le 3$:
	$$\Pr\left[\mathcal{E}_{1}\right] \ge \Pr\left[\mathcal{E}_{2}\right] \ge 1 - \frac{1}{poly(n)}.$$
	where $\mathcal{E}_{1}$ is the event that $MST_i(\bar w_k) \le (1 + \eta) MST_i(w)$	and $\mathcal{E}_{2}$ is the event that $MST_i(\bar w_k) \le (1 + 2 \gamma)(1 + \delta) MST_i(w)$.
	
	After $\bar w_k$ is constructed by running Boruvka's algorithm on it we find an MST exactly and hence the approximation guarantee for each of the MST edges follows.
\end{proof}

\subsection{Exact Hamming MST}\label{sec:exact-hamming-mst}

\begin{theorem}\label{thm:hamming-k-slc}
	For $d = O(1)$ Hamming MST can be computed exactly in $O(\log n)$ rounds of MPC.
\end{theorem}
\begin{proof}
	For the special case $d = 2$ the algorithm is particularly simple and is given in Section~\ref{sec:2d-hamming-simple}.
	For $d = O(1)$ we construct an auxiliary graph with integer edge weights in the interval $[1, \dots, d]$ and use MPC algorithm graph connectivity to compute MST for it.
	
	For a binary vector $b \in \{0,1\}^d$ of Hamming weight $t$ and $v \in \mathbb R^d$ we use notation $v(b)$ to denote a vector in $\mathbb R^t$ consisting of coordinates of $v$ corresponding to non-zero values of $b$.
	For a binary vector $b \in \{0,1\}^d$ we define an order relation $\le_b$ on vectors in $\mathbb R^d$ as follows: $v_1 \le_b v_2$ if and only if $v_1(b) \le v_2(b)$ and $\le$ is the lexicographic order. 
	\begin{itemize}
		\item Initialize $G$ to an empty graph on $n$ vertices and $\mathcal F$ to a forest of singleton vertices.
		\item For each binary vector $b \in \{0,1\}^d$:
		\begin{itemize}
			\item Sort input vectors $v_1, \dots v_n$ according to $\le_b$ breaking ties arbitrarily.
			\item Let $\pi(i)$ be the index of the $i$-th vector in this sorted order.
			\item For $i = 1, \dots, n - 1$ create an edge in $G$ of weight $d - \|b\|_0$ between vertices $\pi(i)$ and $\pi(i + 1)$ if $v_{\pi(i)}(b) = v_{\pi(i + 1)}(b)$
		\end{itemize}
		\item For $i = 1, \dots, d$
		\begin{itemize}
			\item Augment $\mathcal F$ using edges of length $i$ in $G$ to a spanning forest for the subgraph of $G$ consisting of edges of weight at most $i$.
		\end{itemize}
		
	\end{itemize}
	
	The first loop of the above reduction can be performed in $O(1)$ rounds of MPC by replicating the data $2^d = O(1)$ times and running $O(1)$-round MPC sorting algorithm~\cite{GSZ11} on all the replicas in parallel. The second loop can be performed in $O(d \log n)$ rounds total by using $O(\log n)$ connectivity algorithm in each iteration.
	Correctness of the above algorithm follows from the following observation:  if $\|v_i - v_j\|_0 = t$ then there exists a path in the graph $G$ between $i$ and $j$ that uses only edges of weight at most $t$.

	Indeed, if Hamming distance between two vectors equals $t$ then there exists a subset of $d - t$ coordinates where these two vectors agree. Let $b \in \{0,1\}^d$ be the indicator vector of this subset. In the iteration of the first loop corresponding to $b$ let $p_i$ and $p_j$ be positions of vectors $v_i$ and $v_j$ in the sorted order in this iteration. W.l.o.g $p_i < p_j$ and for all $k = p_i, \dots, p_j - 1$ we added an edge of weight $t$ between $\pi(k)$ and $\pi(k + 1)$ creating the desired path in $G$.
	
	From the above observation it follows directly that for all $t =1, \dots, d$ the number of connected components in the subgraph of $G$ induced by edges of weight at most $t$ is the same as the number of connected components induced by edges of Hamming weight at most $t$ in the original input. Thus executions of Kruskal's algorithm on $G$ and the distance graph under Hamming distance give the same result and hence the MST constructed by the algorithm above is optimal.

	\subsubsection{Simple proof of Theorem~\ref{thm:hamming-k-slc} for $d = 2$}\label{sec:2d-hamming-simple}
	
	An instance of Hamming MST for $d = 2$ is represented by $n$ vectors $(x_1, y_1), \dots, (x_n, y_n)$.
	Because all edges in the distance graph have cost either $1$ or $2$ the cost of the optimum MST equals to $n + c - 2$ where $c$ is the number of connected components in the subgraph induced by edges of cost $1$.
	We will construct a subgraph that has the same set of connected components using only $2n$ edges and then run an $O(\log n)$-round MPC connectivity algorithm on it.
	
	Formally the construction is given as follows. First, we create a vertex $i$ for each input vector $(x_i, y_i)$.
	Then we create edges between these vertices as described below, repeating this process with the role of $x$ and $y$ coordinates flipped (i.e. sort and group according to $y$). 
	\begin{itemize}
		\item Sort input vectors according to the $x$-coordinate and then according to the $y$-coordinate.
		\item For each value of $x$ let $y^x_1, \dots, y^x_t$ be the corresponding sorted $y$-coordinate values.
		\item For all $j$  where $1 \le j < t$ create an edge between vertices representing $(x, y^x_j)$ and $(x, y^x_{j + 1})$ in the input graph.
	\end{itemize}
\end{proof}
This reduction can be performed in a constant number of rounds of MPC using a constant-round MPC sorting algorithm of ~\cite{GSZ11}. Furthermore, it preserves the connected components in the graph induced by edges of length $1$ under Hamming distance.
Indeed, such edges correspond to pair of vectors that have one of the coordinates being equal and hence in our construction there is a path between vertices representing such vectors. This completes the proof for the case $d = 2$.

\section{Hardness of $k$-SLC}\label{sec:hardness}
\subsection{Hardness under $\ell_1$ and $\ell_2$}

The following two conjectures are widely used in the MPC literature~\cite{KSV10, BKS13,RMCS13,RVW16}. Note that the second conjecture is stronger and hence can potentially be used to get stronger hardness results.
\begin{conj}[Sparse connectivity hardness]\label{conj:connectivity}
	If $s = n^{\alpha}$ for a constant $\alpha < 1$ then solving connectivity on an input graph with $n$ vertices and $O(n)$ edges requires $\Omega(\log n)$ rounds of MPC.
\end{conj}
\begin{conj}[One cycle vs. two cycles hardness]\label{conj:cycles}
	If $s = n^{\alpha}$ for a constant $\alpha < 1$ then distinguishing the following two instances requires $\Omega(\log n)$ rounds of MPC: 1) a cycle on $n$ vertices, 2) two cycles on $n/2$ vertices each.
\end{conj}

%\subsection{$k$-SLC under $\ell_1$ and $\ell_2$}

\begin{theorem}\label{thm:l1l2-k-slc-hardness}
	No $o(\log n)$-round MPC algorithm can achieve approximation for $2$-SLC:
	\begin{enumerate}
		\item Better than $(\sqrt{2 + \sqrt{2}}-\epsilon)$ under $\ell_2$ for $d = \Omega(\log n/\epsilon^2)$ under Conjecture~\ref{conj:cycles}.
		\item Better than $3$ under $\ell_1$ for $O(1)$-sparse vectors and $d = \Omega(n)$ under Conjecture~\ref{conj:cycles}.
		\item Better than $(\sqrt{2} - \epsilon)$  under $\ell_2$ for $d = \Omega(\log n/\epsilon^2)$ under Conjecture~\ref{conj:connectivity}.
		\item Better than $2$ under $\ell_1$ for $O(1)$-sparse vectors and $d = \Omega(n)$ under Conjecture~\ref{conj:connectivity}.
	\end{enumerate}
\end{theorem}
\begin{proof}
	\textbf{Part 1.} Given an instance of the ``one cycle vs. two cycles problem'' we reduce it to the $2$-SLC problem as follows: 
	\begin{enumerate}
		\item Create a vector $v'_i \in \mathbb R^n$ for each vertex where $v'_i = e_i$ and $e_i$ is the $i$-th standard unit vector.
		\item For each edge $(a,b)$ in the input graph update the corresponding vectors as $v'_a = v'_a + \xi e_b$ and $v'_b = v'_b + \xi e_a$ where $\xi = \frac{1}{\sqrt{2}}$.
		\item Apply Johnson-Lindenstrauss transform to $v'_1, \dots, v'_n$ to construct $v_1, \dots, v_n \in \mathbb R^d$ where $d = O(\log n/\epsilon^2)$.
	\end{enumerate}
	It is important that the above reduction can be performed in only a constant number of MPC rounds. Indeed, Step 1 can be done locally by partitioning vectors between machines and to perform Step 2 we can send each edge $(a,b)$ to the machines holding vectors $v_a$ and $v_b$.
	In order to perform Step 3 note that for each $i$ we have $v_i = M v'_i$ where $M$ is the Johnson-Lindenstrauss matrix and each $v'_i$ has at most $3$ non-zero entries. Hence, all $v_i$ can be computed in one round of MPC with only $O(\log n/\epsilon^2)$ communication per vector.
	
	\begin{proposition}\label{prop:max-ratio}
		If there is an edge $(i,j)$ in the input graph then $\|v'_i - v'_j\|_2 = \sqrt{2}(\sqrt{2 - \sqrt{2}})$, otherwise $\|v'_i - v'_j\|_2 = 2$.
	\end{proposition}
	\begin{proof}
		Indeed, if there is an edge $(i,j)$ in the input then there exist two other edges $(i,i')$ and $(j,j')$ and hence, the non-zero entries of $v'_i$ and $v'_j$ are as follows: $v'_{ii} = 1, v'_{ii'} = \xi, v'_{ij} =\xi, v'_{jj} = 1, v'_{jj'} = \xi, v'_{ji} = \xi$. Hence $\|v'_i - v'_j\|_2 = \sqrt{2(1 - \xi)^2 + 2\xi^2}$.
		On the other hand, if there is no edge $(i,j)$ then there exist four edges $(i,i'), (i,i''), (j,j')$ and $(j,j'')$ and non-zero entries of $v'_i$ and $v'_j$ are:
		$v'_{ii} = 1, v'_{ii'} = \xi, v'_{ii''} =\xi, v'_{jj} = 1, v'_{jj'} = \xi, v'_{jj''} = \xi$. Hence $\|v'_i - v'_j\|_2 = \sqrt{2 + 4\xi^2}$. Maximum of the ratio $\frac{\sqrt{2 + 4\xi^2}}{\sqrt{2(1 - \xi)^2 + 2\xi^2}}$ is achieved when $\xi = 1/\sqrt{2}$ and equals $\sqrt{2 + \sqrt{2}}$.
	\end{proof}
	Clearly by Proposition~\ref{prop:max-ratio}, if the input graph is one cycle then the cost of $2$-SLC of $v'_1, \dots, v'_n$ equals $\sqrt{2}\sqrt{2 - \sqrt{2}}$, otherwise it is $2$. As Johnson-Lindenstrauss transform preserves all pairwise distances up to a multiplicative $(1\pm\epsilon)$ factor with high probability the same is true for the cost of $2$-SLC of $v_1, \dots, v_n$ up to $\pm \epsilon$ error. This completes the proof.
	
	\textbf{Part 2.} We perform reduction as in Part 1 but without Step 3 and setting $\xi = 1$.
	Note that since the resulting vectors have at most 3 non-zero entries each the input can be represented in $O(n)$ space.
	A calculation similar to the above shows that in this case if there is an edge $(i,j)$ in the graph then $\|v'_i - v'_j\|_1 = 2|1 - \xi| + 2 |\xi|$. Otherwise, $\|v'_i - v'_j\|_1 = 2 + 4 |\xi|$.
	The ratio between these two cases is maximized when $\xi = 1$ and equals $3$.
	\qedhere
	
	\textbf{Part 3.} Given an instance $G(V,E)$ of sparse connectivity we reduce it to $\ell_2$-2-SLC as follows.
	Let $n = |V|$ and $m = |E|$.
	We can assume that $G$ has no isolated vertices as connectivity instances containing such vertices can be solved in $O(1)$ rounds of MPC by identifying isolated vertices.
	
	For the $i$-th edge of the input we create a vector $v'_i \in \mathbb R^n$. We set $v'_{i,j} = 1$ if the $i$-th edge in the input is adjacent on vertex $j$ and $v'_{ij} = 0$ otherwise. Then we apply JL-transform to reduce the dimension of constructed vectors to $O(\log n/\epsilon^2)$ obtaining vectors $v_1, \dots, v_m$ as in Part 1. 
	
	Since $G$ has no isolated vertices it is connected if and only if the set of its edges forms a connected subgraph. By construction if two edges $i$ and $j$ share a vertex then $\|v'_i - v'_j\| = \sqrt{2}$, otherwise $\|v'_i - v'_j\| = 2$.
	Hence, if $G$ is connected all edges in the $\ell_2$-MST for $v_1, \dots, v_m$ have length $\sqrt{2}$. Otherwise, there exists an edge in $\ell_2$-MST of length $2$.
	Hardness of $(\sqrt{2}-\epsilon)$-approximation hence follows from the fact that JL-transform preserves all distances up to $(1 \pm \epsilon)$-approximation.
	
	\textbf{Part 4.} We use the same reduction as in Part 3 but without using the JL-transform in the end. Since the instance of sparse connectivity has $O(n)$ edge we obtain $O(n)$ vectors with $2$ non-zero entries in each. Hence, resulting instance can be stored in $O(n)$ space. As in Part 3 note that if two edges $i$ and $j$ share a vertex then $\|v_i - v_j\|_1 = 2$, otherwise $\|v_i - v_j\| = 4$. 
	Hence the costs of $\ell_1$-$2$-SLC differ by a factor of $2$ depending on whether $G$ is connected or not.
	
\end{proof}

\subsection{Hardness of Hamming $k$-SLC}\label{sec:hamming-hardness}

\begin{theorem}\label{thm:hamming-k-slc-hardness}
	No algorithm for computing Hamming $k$-SLC	cost for $d = 2$ in $o(\log n)$ rounds of MPC can achieve better than $2$-approximation under Conjecture~\ref{conj:connectivity}.
\end{theorem}
\begin{proof}
	Let $G(V,E)$ be an instance of sparse connectivity. Our reduction to Hamming $2$-SLC constructs an input set of $2$-dimensional vectors as follows:
	\begin{itemize}
		\item For each vertex $i \in V$ create a vector $(i,i)$.
		\item For each edge $(i,j) \in E$ create a vector $(i,j)$.
	\end{itemize}
	Clearly this reduction can be performed in a constant number of rounds of MPC and the resulting instance has $|V| + |E| = O(n)$ many vectors.
	We will show that if the input graph is connected the cost of Hamming $2$-SLC of the input equals $1$ and the cost is $2$ otherwise. 
	Indeed, note that the distances between resulting vectors are always either $1$ or $2$.
	If $G$ is connected then it is easy to construct a connected spanning subgraph in the resulting Hamming graph where each edge has cost $1$.
	Indeed, consider a subgraph that for each edge $(i,j)$ in the input graph contains two edges: one between vectors $(i,i)$ and $(i,j)$ and another between vectors $(j,j)$ and $(i,j)$.
	Clearly, if the input graph is connected then this is a connected spanning subgraph. Hence the Hamming MST cost of the constructed point set equals $|V| + |E| - 1$ and the Hamming $2$-SLC cost equals 1.
	On the other hand, if $G$ is disconnected then consider any partitioning $(S,T)$ of $G$ into connected components.
	Clearly, any two vectors representing vertices belonging to different parts of this partition in our reduction are at distance $2$ from each other. This implies that the Hamming MST cost is at least $|V| + |E|$ and the Hamming $2$-SLC cost is $2$.
	
\end{proof}

\section{Implementation details and performance analysis}\label{sec:implementation}
The implementation details in this section are similar to~\cite{ANOY14a}. The main difference is that we need the Solve-and-Sketch to work under $\ell_1$ and $\ell_\infty$ rather than just $\ell_2$ as in~\cite{ANOY14a} which requires some modifications in the analysis.

\subsection{Distance-preserving partitions}

We start by recalling definitions of distance preserving hierarchical partitions from~\cite{ANOY14b}.
Let $M(S,\rho)$ be a metric space with distance function $\rho$.
For $S' \subseteq S$ we denote its diameter as $\Delta(S') = \sup_{x,y \in S'} \rho(x,y)$.
A \textit{deterministic hierarchical partition} P with $L$ \textit{levels} is defined as a sequence $P = (P_0, \dots, P_L)$ where $P_L = \{S\}$ and each level $P_\ell$ is a subdivision of $P_{\ell + 1}$. For a partition $P_i$ we call its parts \textit{cells}. The \textit{diameter} at level $i$ is defined as $\Delta(P_i) = \max_{C \in P_i} \Delta(C)$.
The \textit{degree of a cell} $C \in P_\ell$ is $deg(C) = |\{C' \in P_{\ell -1}: C' \subseteq C\}|$.
The \textit{degree of a hierarchical partition} is the maximum degree of any of its cells. The unique cell at level $\ell$ containing a point $x$ is denoted as $C_\ell(x)$. We say that a partition is indexable if this cell can be computed based on $x$ and $\ell$.
A \textit{randomized hierarchical partition} is a distribution over deterministic hierarchical partitions.

\begin{definition}[Distance-preserving partition]~\cite{ANOY14b}
	For parameters $a \in (0,1)$, $b,c \in \mathbb R^+$ and $\gamma > 1$ a randomized hierarchical partition $\mathcal P$ of a metric space with $L$ levels is $(a,b,c)$-distance-preserving with approximation $\gamma$ if the degree of all deterministic partitions in its support is at most $c$ and the following properties are satisfied for $\Delta_\ell = \gamma a^{L - \ell}\Delta(S)$:
	\begin{enumerate}
		\item (Bounded diameter) For every deterministic partition $P = (P_0, \dots, P_L)$ in the support of $\mathcal P$ and for all $\ell \in \{0, \dots, L\}$ it holds that:
		$$\Delta(P_\ell) \le \Delta_\ell.$$
		\item (Probability of cutting an edge) For every $x,y \in S$ and for all $\ell \in \{0, \dots, L\}$:
		$$\Pr_{P \sim \mathcal P}[C_\ell(x) \neq C_\ell(y)] \le b \frac{\rho(x,y)}{\Delta_\ell}.$$
	\end{enumerate}
\end{definition}

We use the following construction of~\cite{ANOY14b} to build such a distance-preserving partition $\mathcal P$ for $S \subseteq \mathbb R^d$.
We can always shift $S$ such that all points fit into a box $[0, \Delta]^d$ where $\Delta$ is the diameter of the metric space$(S,\ell^d_\infty)$.
Pick a vector $r \in \mathbb R^d$ uniformly at random from $[0,\Delta]^d$.
Two points $u$ and $v$ belong to the same cell at level $\ell \in \{0, \dots, L\}$ if and only if for all dimensions $i \in [d]$ it holds that $\lfloor \frac{(u_i - r_i)\alpha^{L - \ell}}{\Delta}\rfloor = \lfloor \frac{(v_i - r_i)\alpha^{L - \ell}}{\Delta}\rfloor$ where $\alpha$ is a parameter (see Figure~\ref{fig:abc-partition} for an example).
Note that this partition is indexable since coordinates of the point $x \in \mathbb R^d$, random shift $r$ and $\ell$ suffice for computing $C_\ell(x)$.

\begin{lemma}[\cite{ANOY14b}, Lemma 5.3]\label{lem:l2-partition}
	Indexable randomized hierarchical partition $\mathcal P$ given by the construction above has $L = O(\log_a |S|)$ levels can be constructed in $O(1)$ rounds of MPC and is an $(1/\alpha, d, (\alpha + 1)^d)$-distance-preserving partition for $(S, \ell^d_2)$ with approximation $\gamma \le \sqrt{d}$.
\end{lemma}

Below we show that this partition $\mathcal P$ is also distance-preserving for $\ell^d_1$ and $\ell^d_\infty$.

\begin{lemma}\label{lem:l1linf-partition}
	Indexable randomized hierarchical partition $\mathcal P$ is: 1)  $(1/\alpha, d^2,(\alpha + 1)^d)$-distance-preserving for $\ell^d_1$ with approximation $d$,
	2) $(1/\alpha, d, (\alpha + 1)^d)$-distance-preserving for $\ell^d_\infty$ with approximation $1$.
\end{lemma}
\begin{proof}
	The degree bound of $(\alpha + 1)^d$ follows from Lemma~\ref{lem:l2-partition} in both cases so in the rest of the proof we only analyze other parameters of the partition.
	
	\textbf{Part 1.}
	Under $\ell_1^d$ we have $\Delta(S) \ge \Delta$ and $\Delta(S) \le d \Delta$.
	By construction cells of the partition at level $\ell$ have diameter at most $d \Delta \alpha^{\ell - L}$.
	Hence we can set $\gamma = d$ and $\Delta_\ell = d \alpha^{\ell - L}\Delta(S)$ which satisfies the bounded diameter condition $\Delta(P_\ell) \le \Delta_\ell$.
	To verify the condition on probability of cutting an edge consider two points $u,v \in \mathbb R^d$.
	The probability that $\lfloor \frac{(u_i - r_i)\alpha^{L - \ell}}{\Delta}\rfloor \neq \lfloor \frac{(v_i - r_i)\alpha^{L - \ell}}{\Delta}\rfloor$ for a fixed $i$ is at most $\frac{\|u_i - v_i\|_1 \alpha^{L - \ell}}{\Delta}$.
	By a union  bound the probability that $u$ and $v$ belong to different cells at level $\ell$ is at most 
	$$\frac{\|u - v\|_1 \alpha^{L - \ell}}{\Delta} = \frac{d \|u - v\|_1 \Delta(S)}{\Delta_\ell \Delta} \le \frac{d^2 \|u - v\|_1}{\Delta_\ell}.\qedhere$$
\end{proof}

\textbf{Part 2.}
Under $\ell^d_\infty$ we have $\Delta(S) = \Delta$.
By construction cells of the partition at level $\ell$ have diameter at most $\Delta \alpha^{\ell - L}$. Hence we can set $\gamma = 1$ and $\Delta_\ell = \alpha^{\ell - L} \Delta$.
As in the previous case the probability that two vectors $u, v \in \mathbb R^d$ belong to different cells of the partition is at most $\frac{\|u - v\|_1 \alpha^{L - \ell}}{\Delta}$ which can be bounded as follows:
$$\frac{\|u - v\|_1 \alpha^{L - \ell}}{\Delta} \le \frac{d \|u - v\|_\infty \alpha^{L - \ell }}{\Delta} = \frac{d \|u - v\|_\infty}{\Delta_\ell}.$$
Thus $\mathcal P$ is indeed an $(1/\alpha, d, (\alpha + 1)^d)$-distance-preserving partition for $\ell^d_\infty$ with approximation $1$.
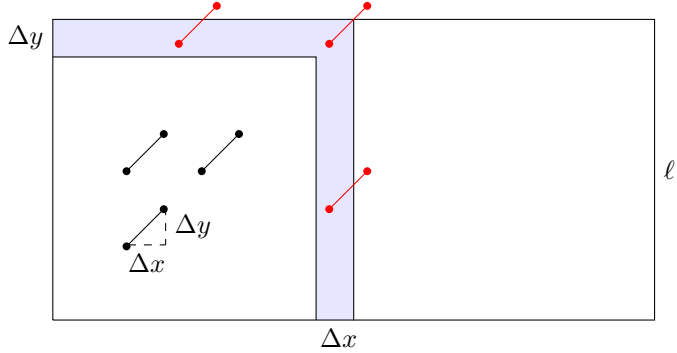
\begin{figure}
	\begin{center} 
	\usetikzlibrary{shapes.geometric}
	\begin{tikzpicture}
	\fill[blue!10!white] (3.5,0) rectangle (3.98,3.98);
	\fill[blue!10!white] (0,3.5) rectangle (3.98,3.98);
	\draw (0,0) -- node[pos= .95,below] {$\Delta x$} (4,0) -- (4,4) -- (0,4) --  node[pos= 0.06,left] {$\Delta y$}  (0,0) ;
	\draw (4,0) -- (8,0) --  node[pos= 0.5,right] {$\ell$}(8,4) -- (4,4);
	\draw (3.5, 0) -- (3.5, 3.5) -- (0, 3.5);
	
	\draw[red] (1.7, 3.7) -- (2.2,4.2);
	\fill[red] (1.675, 3.675) circle (1.5pt);
	\fill[red] (2.180, 4.180) circle (1.5pt);
	
	\draw[red] (3.7, 3.7) -- (4.2,4.2);
	\fill[red] (3.675, 3.675) circle (1.5pt);
	\fill[red] (4.180, 4.180) circle (1.5pt);
	
	\draw[red] (3.7, 1.5) -- (4.2,2.0);
	\fill[red] (3.675, 1.475) circle (1.5pt);
	\fill[red] (4.180, 1.980) circle (1.5pt);

	\draw (2.5, 2.5) -- (2.0,2.0);
	\fill[black] (2.475, 2.475) circle (1.5pt);
	\fill[black] (1.980, 1.980) circle (1.5pt);

	\draw (1.5, 1.5) -- (1.0,1.0);
	\fill[black] (1.475, 1.475) circle (1.5pt);
	\fill[black] (0.980, 0.980) circle (1.5pt);
	
	\draw[dashed] (1.5, 1.5) -- node[pos= 0.5,right] {$\Delta y$} (1.5, 1.0);
	
	\draw[dashed] (1.0, 1.0) -- node[pos= 0.5,below] {$\Delta x$} (1.5, 1.0);

	\draw (1.5, 2.5) -- (1.0,2.0);
	\fill[black] (1.475, 2.475) circle (1.5pt);
	\fill[black] (0.980, 1.980) circle (1.5pt);
	\end{tikzpicture}
	
	\caption{Probability of a cut proportional to $\Delta x + \Delta y$ corresponds to edge landing in the shaded region}
	\label{fig:abc-partition}
\end{center}
\end{figure}

\subsection{Solve-and-Sketch framework}
\label{sec:solve-and-sketch}
%\begin{theorem}[Theorem 5.2,~\cite{ANOY14b}]
%Let $P$ be an $(a,b,c)$-distance-preserving partition \todo{det or rand?}with $L$ levels of the input $S$ labeled by $P$.\todo{expalin what $s$ is}
%Let $\mathcal A_u$ be a unit step \todo{introduce unit step} algorithm which for input of size $n_u$ has time $t_u = t_u(n_u)$, space $s_u = s_u(n_u)$ and output of size $p_u = p_u(n_u)$. Assume all functions are non-decreasing and let $p_u=p_u(s)$, $t_u = t_u(c p_u)$. 
%
%If $s_u(c p_u) \le s$ and $p_u \le \frac{\sqrt{s}}{4cL}$ then Algorithm 2 in~\cite{ANOY14b} can be executed in $R = O(\log_s n)$ rounds in the MPC model. Furthermore, after $R$ rounds this algoirthm will apply the unit step to all cells of the hierarchical partition $P$. If comparing two cells in some good ordering \todo{introduce good ordering and specify comparison time} requires time $\tau$, then local computation time per machine in a round is bounded by $O(s \cdot (\tau \cdot \log s \log_s n) + L t_u)$. All bounds hold with high probability.
%\end{theorem}	
%
%In~\cite{ANOY14b} it is shown that for $\ell_2$ one can achieve

\newcommand{\ccomp}{Q}
\newcommand{\supp}{supp}
\newcommand{\weight}{w_P}
\newcommand{\weighti}[1]{\weight^{#1}}
\newcommand{\mycell}{C}
\newcommand{\randpart}{\mathcal P}
\newcommand{\detpart}{P}
\newcommand{\myparagraph}[1]{\smallskip\noindent{\bf #1}}
\newcommand{\comp}[1]{c(#1)}
\newcommand{\cost}[1]{\rho(#1)}
\newcommand{\copt}{\mathcal C^*}
\newcommand{\cour}{\mathcal C}
\newcommand{\tropt}{\mathcal T^*}
\newcommand{\trour}{\mathcal T}
\newcommand{\boxapprox}{\gamma}
\newcommand{\lev}{L}
\newcommand{\levdiam}{\Delta}

We use Solve-and-Sketch (SAS) framework of~\cite{ANOY14a} for computing an approximate
minimum spanning tree.  SAS framework works with a partition $\detpart =
(\detpart_0, \ldots, \detpart_\lev)$ of the input $M(S,\rho)$, sampled
from a randomized $(a, b, c)$-partition $\randpart$. Then SAS
algorithm proceeds through $\lev$ levels, and in level $\ell$ a
\emph{unit step} algorithm $\mathcal A_u$ is executed in each cell $\mycell$ of the partition $\detpart_\ell$, with input the union of the outputs of the
unit steps applied to the children of $\mycell$. The unit step
also outputs a subset of the edges of a spanning tree in addition to
the input for the next level. 
Once the unit step has been executed for the root cell of partition at level $P_L$ (and hence also for all other cells) the computation is complete.
We give the description of the unit step algorithm of ~\cite{ANOY14a} below.

\begin{definition}[$\delta$-covering]
	Let $M=(S,\rho)$ be a metric space and let $\delta > 0$ .
	A set $S' \subseteq S$ is a \emph{$\delta$-covering} if for any point $x \in S$, there is a point $y \in S'$ such that $\rho(x,y) \le \delta$.
\end{definition}

\begin{algorithm}
	\caption{Unit Step at Level $\ell$,}\label{alg:unit-step} 
	\SetKwInOut{Input}{input}\SetKwInOut{Output}{output} 
	
	\Input{Cell $\mycell \in \detpart_\ell$, a collection $V(\mycell)$
		of points in $\mycell$, and a partition $\ccomp = \{\ccomp_1,
		\ldots \ccomp_k\}$ of $V(\mycell)$ into previously computed
		connected components.}
	
	$\theta := 0$\\
	\While{$k>1$ and $\theta \le \epsilon \levdiam_\ell$} {
		Let $\tau = \min_{\substack{i, j\\i\ne j}} \min_{u \in \ccomp_i, v \in \ccomp_j} \rho(u,v)$\\
		Find $u \in \ccomp_i$ and $v \in \ccomp_j$ for some $i$ and $j$ such that $i\ne
		j$ and $\rho(u,v) \leq
		(1+\epsilon)\tau$.\\
		$\theta := \rho(u,v)$\\
		\If{$\theta \le \epsilon \levdiam_\ell$}{
			Output tree edge $(u,v)$.\\
			Merge $\ccomp_i$ and $\ccomp_j$ and update $\ccomp$ and $k$. \\
	} }
	\Output{$V' \subseteq V$, an $\epsilon^2
		\levdiam_\ell$-covering for $\mycell$, the partition $\ccomp(V')$
		induced by $\ccomp$ on $V'$. }
\end{algorithm}

Below we first show in Lemma~\ref{lem:mpc-implementation} that if the unit step can be executed efficiently then the overall computation also can. 
This theorem is analogous to Theorem 5.2 in~\cite{ANOY14b} but here we give a simpler and faster implementation using the fact that all our partitions are indexable.
Then we proceed to describe implementations of Algorithm~\ref{alg:unit-step} for $\ell_1$ and $\ell_\infty$ in Lemma~\ref{lem:unit-step}.
For $\ell_2$ such an implementation is given in~\cite{ANOY14b}, Lemma 3.23 and for $\ell_1$ and $\ell_\infty$ the implementation is analogous using approximate nearest neighbor search~\cite{AMNSW98} and appropriate $\epsilon^2\Delta_\ell$-covering construction for each metric.

\begin{lemma}\label{lem:mpc-implementation}
	Let $t_u(x)$ be a convex function and $s_u(x)$ be a function with at least linear growth, i.e. $s_u(x) \ge x$.
	Let $P = (P_0, \dots, P_L)$ be an indexable partition labeling $M(S,\rho)$ with $L$ levels sampled from a randomized $(a,b,c)$-distance-preserving family $\mathcal P$.
	Let $\mathcal A_u$ be a unit step algorithm which for an input of size $n_u$ takes time at most $t_u(n_u)$, uses space $s_u = s_u(n_u)$ and produces output of size $p_u = O(\min(p, n_u))$ where $p$ is a parameter.
	If space per machine is $s$ and $s_u(c p_u) \le s/3$ then the unit step computations for all cells of the partition can be executed in $O(L)$ rounds of MPC on $O(n/s)$ machines.
	Furthermore, local computation time per machine in every round is bounded by $t_u(s)$.
\end{lemma}
\begin{proof}
	We process the partition level by level assuming that when we process level $k$ all results for level $k - 1$ are already computed.
	Furthermore, all results from the previous round are labeled by cells of the partition at $P_k$ that they belong to.
	Thus in round $k$ we just execute the unit step for all cells in this level and label the results with the cell they belong at level $k + 1$. The latter part can be performed locally using the fact that $P$ is an indexable partition.
	Note that since $p_u = O(\min(p, n_u))$ the overall output produced in each round has size $O(n)$.
	
	\begin{proposition}
		Total number of machines that we need to execute each round is at most $O(n/s)$ and maximum time per machine in a round is at most $t_u(s)$.
	\end{proposition}
	\begin{proof}
		First we estimate total space we need to allocate on machines.
		Let $m_i$ be the number of nonempty subcells of the $i$-th cell in the $k$-th level, i.e. the number of inputs for this cell.
		For each $i$ and $j = 1, \dots, m_{i}$ let $n_{ij}$ be the size of output produced by the $j$-th subcell in the previous round that is now input for the $i$-th cell in $k$-th level. 
		
		We know that $\sum_{i,j} n_{ij} = O(n)$ and $n_{ij} \le p$.
		The input to $\mathcal A_u$ has overall size at most $cp$ and hence each execution uses space at most $s_u(cp) \le s/3$.
		Note that using the fact that $s_u(x) \ge x$ this also implies that the input for each job is of size at most $s/3$.
		
		We assign executions of  $\mathcal A_u$ jobs in round $k$ arbitrarily to machines in such a way that we only start a new machine if there is no existing machine with at least $2s/3$ space available.
		This ensures that in the end of the assignment process each machine has at least $s/3$ unused space for executing the jobs.
		We then execute jobs assigned to each machine sequentially using this space and store all inputs for all jobs locally. Note that our assignment process ensures that if $\mathcal S$ is the total space required to store inputs for all jobs then the total space we use is at most $3 \mathcal S + s = O(n)$ and the total number of machines is at most $3\mathcal S /s + 1 = O(n/s)$.
		
		Suppose we are using $t$ machines in total in this round and let $B_i$ be the subset of jobs assigned to $i$-th machine.
		The maximum time per machine required to execute jobs in $k$-th round is then: 
		
		\begin{align*}
		\max_{i = 1}^t\left[\sum_{j \in B_i} t_u\left(\sum_{\ell = 1}^{m_j} n_{j,\ell}\right)\right] &\le \max_{i = 1}^t \left[t_u\left(\sum_{j \in B_i} \sum_{\ell = 1}^{m_j} n_{j,\ell}\right)\right] \\
		& \le t_u(2s/3) \\
		& \le t_u(s)
		\end{align*}
		
	\end{proof}
	
	\begin{lemma}\label{lem:unit-step}
		If the input metric space $M$ is a subset of $\ell_1^d, \ell_2^d$ or $\ell_\infty^d$, the unit step Algorithm~\ref{alg:unit-step} has space complexity $s_u(n_u) = n_u \log^{O(1)} n_u$ words, time complexity $t_u(n_u) = (d/\epsilon)^{d + 1} n_u \log^{O(1)}n_u$ and output size $p_u = O(\min((1/\epsilon)^{2d}, n_u))$ words.
	\end{lemma}
	
	%For $\ell_1^d$ and $\ell_\infty^d$ we show the following:
	%\begin{lemma}\label{lem:unit-step}
	%	If the input metric space $M$ is a subset of $\ell_p^d$, the unit step Algorithm~\ref{alg:unit-step} has
	%	\begin{enumerate} 
	%		\item for $p= 1$: space complexity $s_u(n_u) = n_u \log ^{O(1)} n_u$ words, time complexity $t_u(n_u) = (d/\epsilon)^{d + 1} n_u \log^{O(1)}n_u$ and output size $p_u = O(\min((1/\epsilon)^{2d}, n_u))$ words.
	%		\item for $p= \infty$: space complexity $s_u(n_u) = s_u(n_u) = n_u \log ^{O(1)} n_u$ words, time complexity $t_u(n_u) = (d/\epsilon)^{d + 1} n_u \log^{O(1)}n_u$ and output size $p_u = O(\min((1/\epsilon)^{2d}, n_u))$ words.
	%	\end{enumerate}
	%\end{lemma}
	\begin{proof}
		For $\ell_2^d$ the proof is given in Lemma 3.23 of~\cite{ANOY14b}.
		For $\ell_1^d$ and $\ell_\infty^d$ the proof is analogous and follows from the fact that we can execute Algorithm~\ref{alg:unit-step} by using approximate nearest neighbor search. Details are analogous to the details for $\ell_2^d$ and are given in~\cite{I00,ANOY14b}. In particular Theorem 3.27 in~\cite{ANOY14b} describes how to use approximate nearest neighbor data structure of~\cite{AMNSW98} for $\ell_2^d$. Since~\cite{AMNSW98} gives data structures with the same performance for $\ell_1^d$ and $\ell_\infty^d$ as well the analysis of time and space performance in these two cases is the same.
		
		The bound on the output size follows from the fact that we can construct an $\epsilon^2\Delta_\ell$-covering by imposing a grid with step size $\xi$ and taking one arbitrary point from each cell of the grid (if it is nonempty).
		Fix $\xi =\epsilon^2 \Delta_\ell/d$ for $\ell_1$ and note that the total number of cells in the grid is at most $(\Delta_\ell/d\xi)^d$.
		Fix $\xi = \epsilon^2 \Delta_\ell$ for $\ell_\infty$ and note that the total number of cells in the grid is $(\Delta_\ell/\xi)^d$.
		In both cases the bound of $(1/\epsilon)^{2d}$ on the size of the output follows.
	\end{proof}
	
	Putting things together we get the proof of Theorem~\ref{thm:k-slc-main}:
	\begin{proof}[Proof of Theorem~\ref{thm:k-slc-main}]
		The algorithm is given as Algorithm~\ref{alg:partition-mst} and hence the approximation guarantee follows from Theorem~\ref{thm:mst-approximation}. Hence it only remains to analyze performance of Algorithm~\ref{alg:partition-mst}.
		Recall that this algorithm uses an $(a,b,c)$-distance-preserving partition where $a = 1/(s^{\alpha/d} - 1)$, $b = poly(d)$ and $c = s^\alpha$.
		Also recall that we set $\epsilon = \min\left(\frac{\eta}{6 c_1 L b}, \frac{\eta}{3 c_2}\right)$ and hence if we set $\kappa = \min\left(\frac{1}{6 c_1 L b}, \frac{1}{3 c_2}\right)$ then $\kappa \eta = \epsilon$ and $\kappa$ is a constant because $c_1, c_2, L$ and $b = poly(d)$ are constants.
		
		A distance-preserving partition in Step~\ref{step:partition} can be constructed  for $\ell_2^d$ using Lemma~\ref{lem:l2-partition} and for $\ell_1^d$ and $\ell_\infty^d$ using Lemma~\ref{lem:l1linf-partition}.
		This step takes $O(1)$ rounds of MPC under the resource constraints of the theorem.
		Thus, it suffices to show that Step~\ref{step:recursion} can be executed with required performance guarantees.
		Note that $p_u = O(\epsilon^{-2d}) = O((\kappa \eta)^{-2d}) \le s^{1 - 2\alpha}$ and hence $c p_u = O(s^{\alpha} s^{1 - 2\alpha}) = O(s^{1 - \alpha})$.
		Using the assumption that $\epsilon, \alpha$ and $d$ are constants by Lemma~\ref{lem:unit-step} we have $s_u(c p_u) = \tilde O(s^{1 - \alpha}) \le s/3$ for sufficiently large $s$.
		Thus combining Lemma~\ref{lem:unit-step} and Lemma~\ref{lem:mpc-implementation} and using the assumption that $\epsilon$ and $d$ are constant  we conclude that Step~\ref{step:recursion} can be executed on $O(n/s)$ machines with local computation time per machine in each round bounded by $\tilde O(s)$.
		Furthermore since the distance-preserving partition used in out algorithm has $O(\log_{s^{\alpha/d}} n) = O(d/\alpha \log_s n) = O(1)$ levels (by Lemma~\ref{lem:l2-partition}) from Lemma~\ref{lem:mpc-implementation} it follows that this step only takes $O(1)$ rounds of MPC. 
		Overall, each iteration of the loop in Algorithm~\ref{alg:partition-mst} can be executed in $O(1)$ rounds of MPC and hence sequential execution takes $O(\log n)$ rounds as desired.
		
		Finally, Boruvka's algorithm in Step~\ref{step:boruvka} is run on $O(n \log n)$ edges and can be implemented in $O(\log n)$ rounds using $\tilde O(n/s)$ machines under constraints of the theorem.
	\end{proof}
	
	%\begin{algorithm}[H]
	%\\
	%\end{algorithm}
	
\end{proof}

\section{Experiments}\label{sec:experiments}
\paragraph{Small datasets} Four standard clustering datasets used in the literature were taken for experimental evaluation: 1) Image dataset, $d = 3$, $n =34112 $ (house images, \url{https://cs.joensuu.fi/sipu/datasets/}), 2) KDDCUP04Bio dataset , $d = 10$, $n = 145751$ (preprocessed to select 10 numerical dimensions out of 74, accessed via the link above), 3) Shuttle data set from the UCI ML repository, $d = 9$, $n = 43500$. 
4) US Census dataset from the UCI ML repository, $d = 8$, $n = 2548285$.

\paragraph{Large datasets}

In order to test scalability we took the largest real-valued vector datasets from the UCI ML repository: SIFT10M and HIGGS.
Both datasets have approximately 11 million entries so constructing the full matrix of distances in memory is clearly infeasible as the size of this matrix would be roughly 960TB in both cases\footnote{Assuming 8-byte double-precision arithmetic.}. Dimension reduction for this data was done using PCA for $d = 3$. Results are given in Table~\ref{tbl:scalability}.
\begin{table}[t]
	\caption{Scalability experiments.}
	\label{tbl:scalability}
	\vskip 0.15in
	\begin{center}
		\begin{small}
			\begin{sc}
				\begin{tabular}{lcccccr}
					\toprule
					Data set & $n$ points & $n^2$ edges & $d$  & Time (s) & $\epsilon$\\
					\midrule
					%					Census & $2.4 \times 10^6$ & $6.0 \times 10^{12}$ & $8$ &  & \\
					SIFT10M & $1.1 \times 10^7$ & $1.2 \times 10^{14}$ & $3$& $1.2\times 10^5$& $3$\\					
					HIGGS & $1.1 \times 10^7$ & $1.2 \times 10^{14}$ &$3$ & $8.4 \times 10^4$ & $10$\\
					\bottomrule
				\end{tabular}
			\end{sc}
		\end{small}
	\end{center}
	\vskip -0.1in
\end{table}
\subsubsection*{Experimental setup}
We implemented Algorithm~\ref{alg:partition-mst} in Java on Apache Spark 2.0.2 for Hadoop 2.7.3.
Experiments were performed on:
\begin{itemize}
	\item \textbf{Google Cloud Dataproc (GCD)} platform on two cluster configurations: 1) single-core 1 master / 7 worker (1m/7w)  cluster, 2) dual-core 1 master / 3 worker (1m/3w) cluster. Each core had an Intel Xeon E5 processor at 2.2--2.6 GHz and 3.75GB RAM + 10GB HDD space. Due to the limitations of the free tier access on GCD the total number of cores in a cluster is limited to 8, which is still sufficient to demonstrate at least and order of magnitude speedup over the benchmark sequential algorithm.
	
	Three standard clustering datasets used in the literature were taken for experimental evaluation: 1) Image dataset, $d = 3$, $n =34112 $ (house images, \url{https://cs.joensuu.fi/sipu/datasets/}), 2) KDDCUP04Bio dataset , $d = 10$, $n = 145751$ (preprocessed to select 10 numerical dimensions out of 74, accessed via the link above), 3) Shuttle data set from the UCI ML repository, $d = 9$, $n = 43500$. These datasets have been normazlied to have coordinates in each dimension have $0$ mean and unit variance.
	Note that under RAM restrictions of our setup building a full graph of $n^2$ distances locally is infeasible. Also, we could have used significantly less RAM per machine but when using GCD the RAM size is fixed for standard instances.
	
	\item \textbf{Local Simulation} with $200$ reducers on a Dell XPS13 Laptop with an Intel core I5 processor and 8GB RAM. 
	
	The dataset used was a standard US Census dataset from the UCI, ML repository, $d = 8$, $n = 2548285$. Once again the datasets were normazlied to have coordinates in each dimension have $0$ mean and unit variance. Due to limitations on the number of cores in the free tier access on GCD, we performed the experiments for this dataset on a local simulation.  	
\end{itemize}

\vspace{30pt}

\begin{center}
	\begin{minipage}{0.49\linewidth}
		\includegraphics[width=\textwidth]{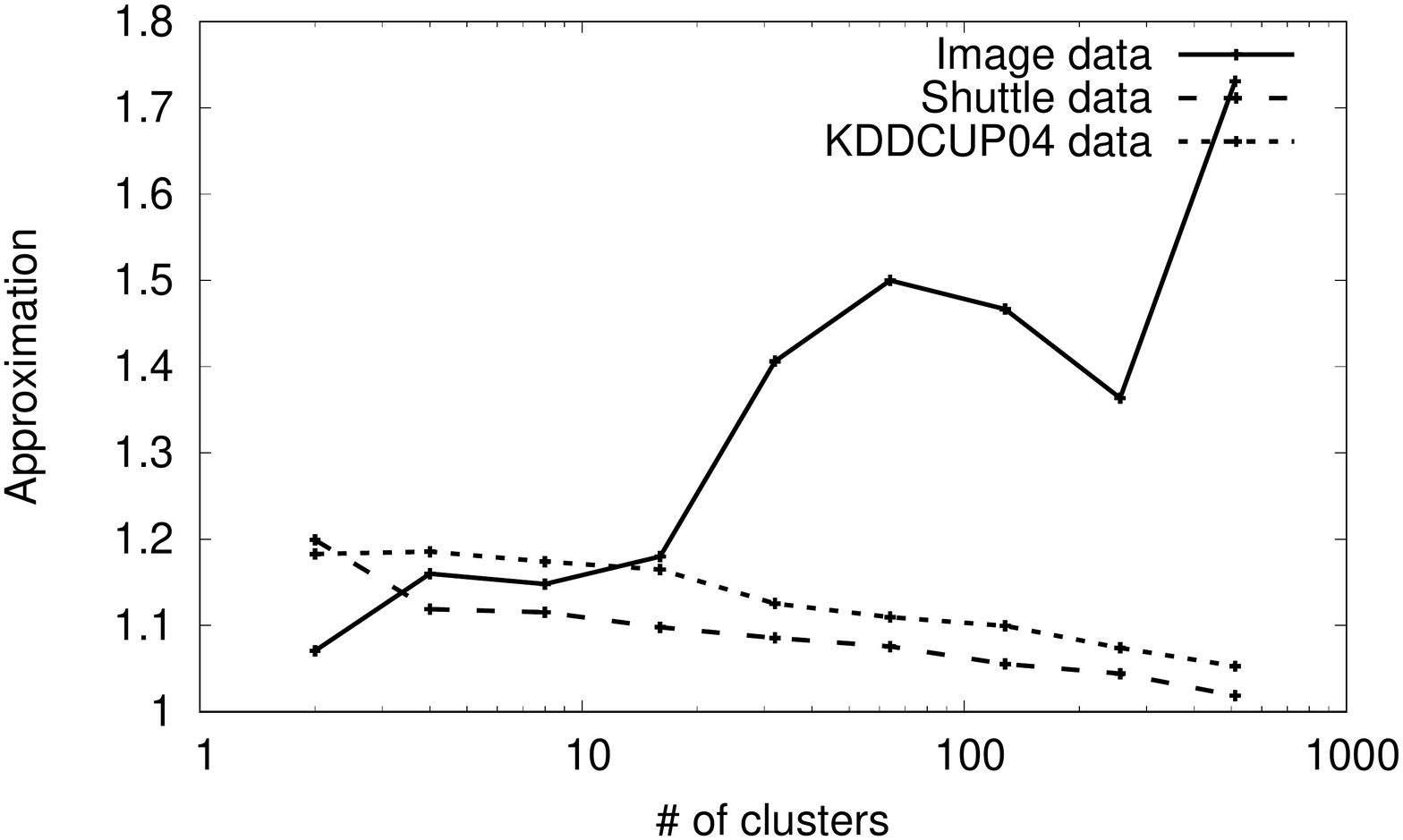}
		\captionof{figure}{Approximation vs. number of clusters}
		\label{plt:approxvsk}
	\end{minipage}
	\begin{minipage}{0.49\linewidth}
		\includegraphics[width=\textwidth]{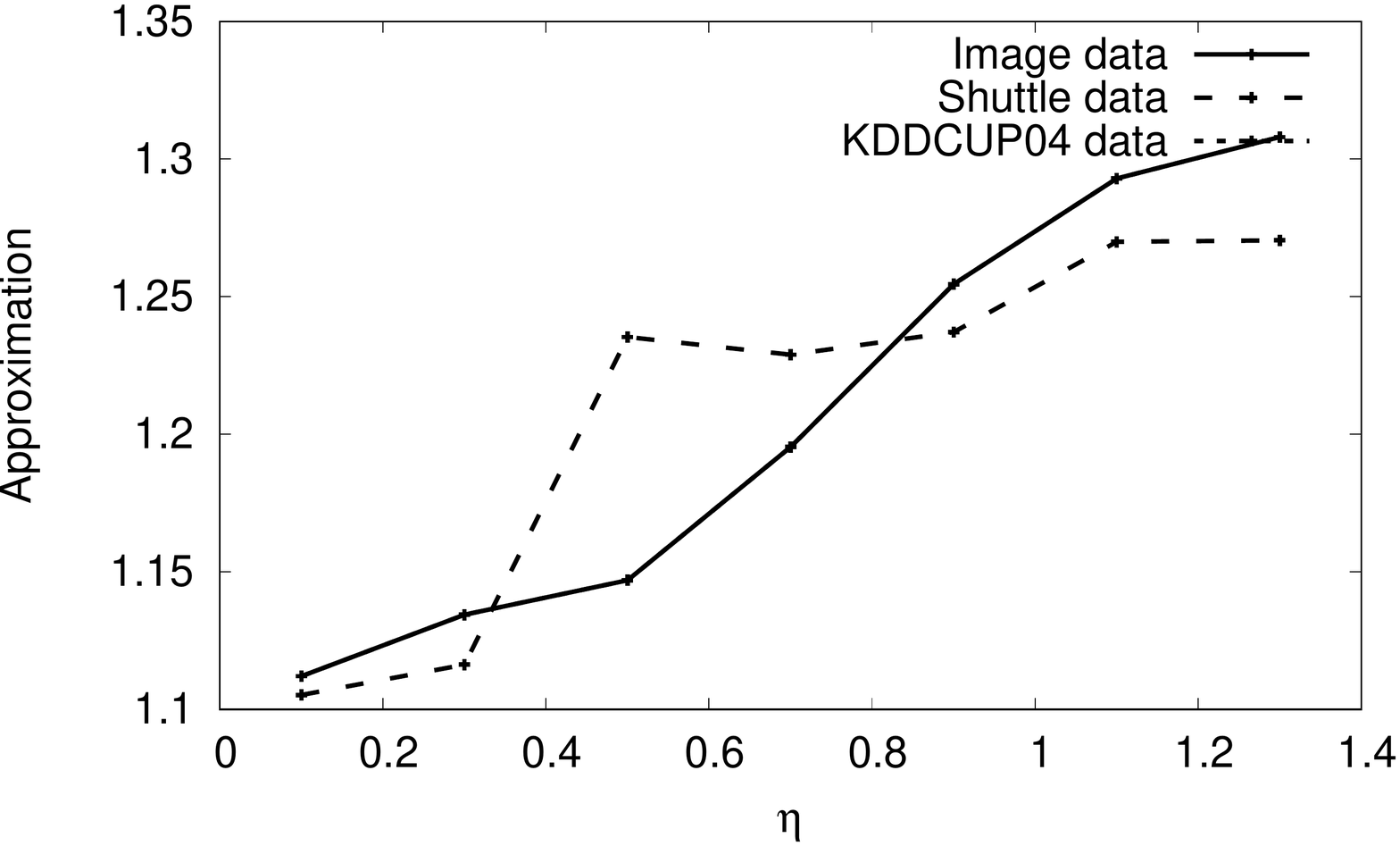}
		\captionof{figure}{Approximation vs. $\eta$ parameter}
		\label{plt:approxvseta}
	\end{minipage}
\end{center}
\begin{center}
	\begin{minipage}{0.49\linewidth}
		\includegraphics[width=\textwidth]{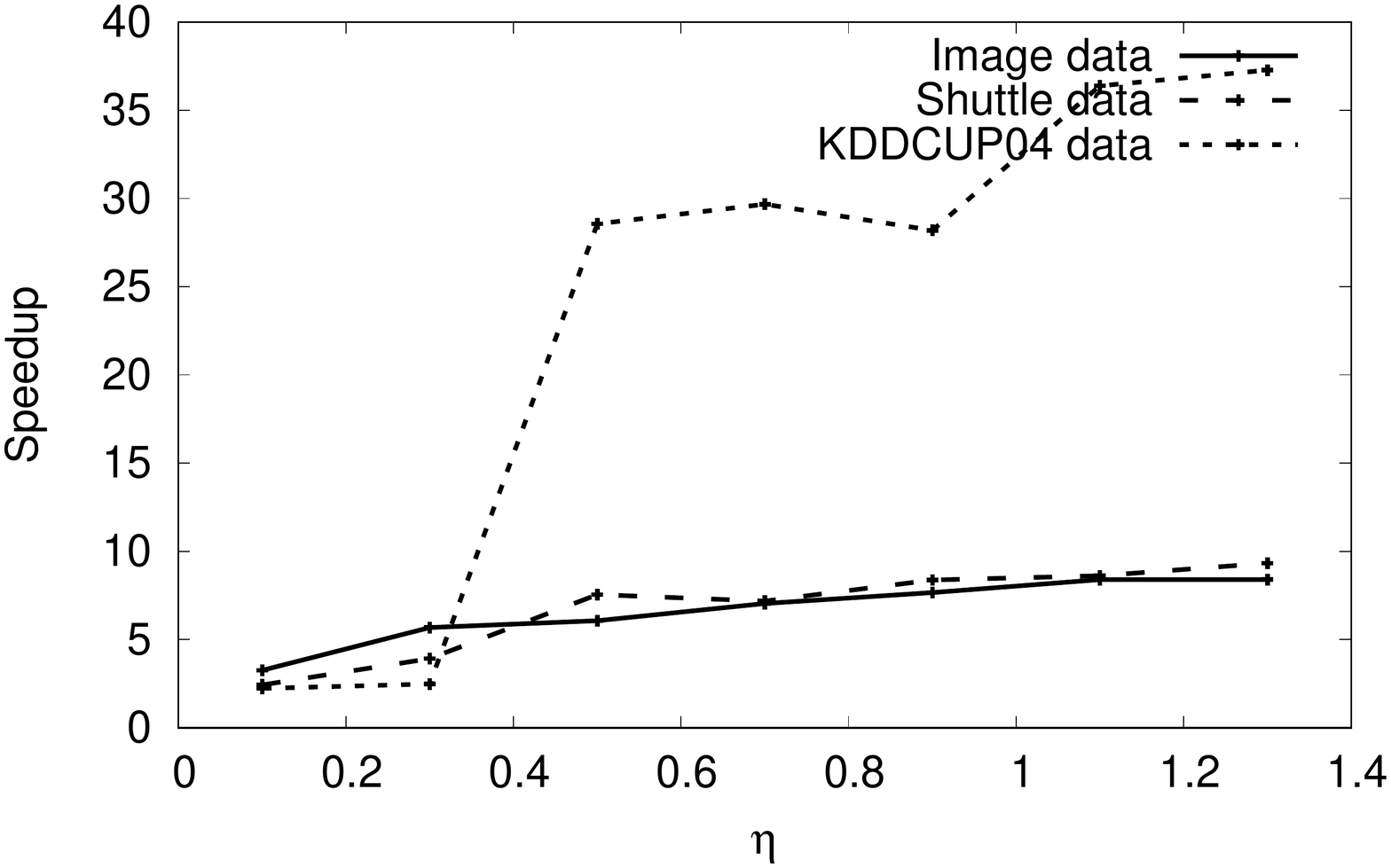}
		\captionof{figure}{Speedup vs $\eta$ parameter, 1m/7w cluster}
		\label{plt:speedupvseta}
	\end{minipage}
	\begin{minipage}{0.49\linewidth}
		\includegraphics[width=\textwidth]{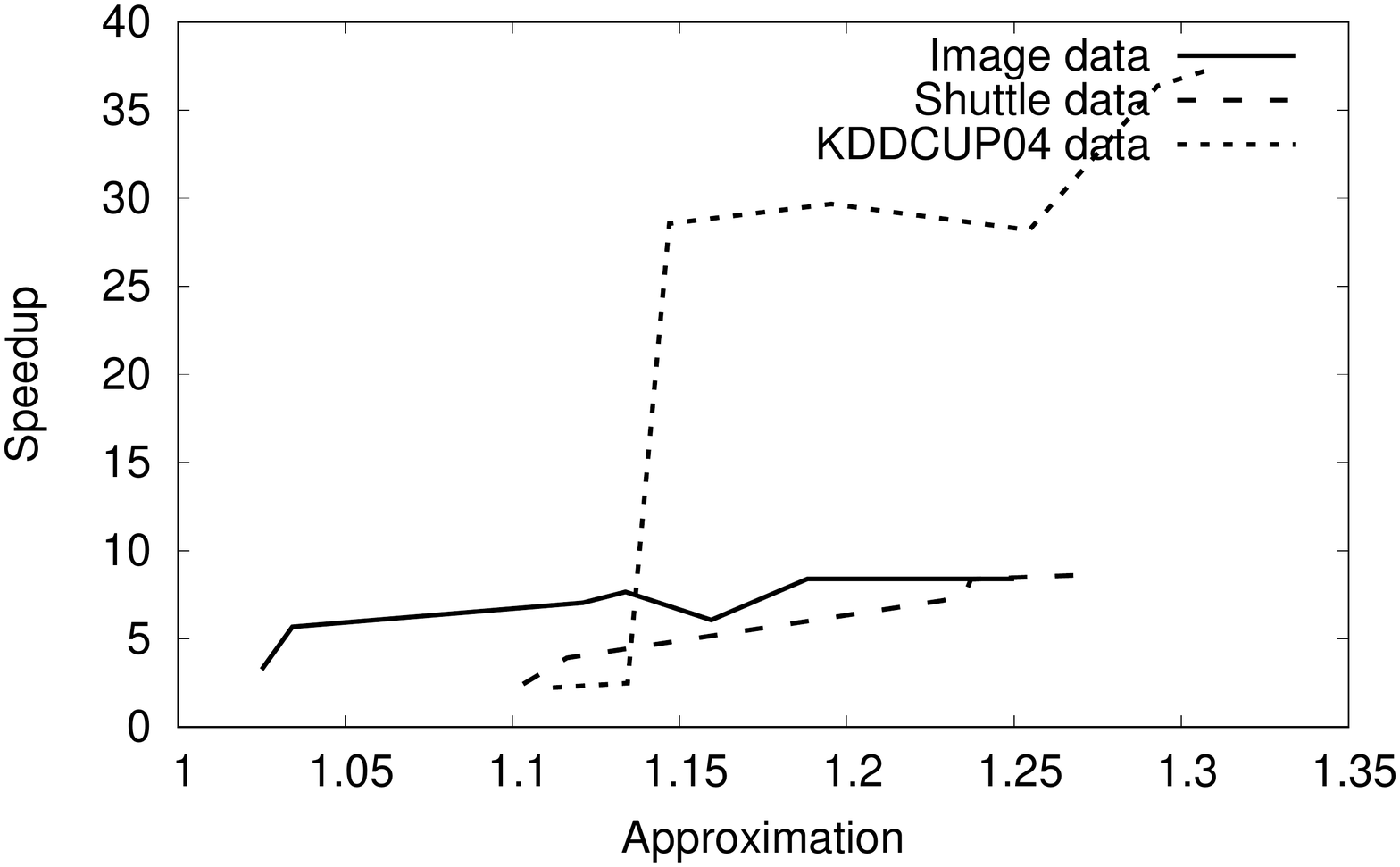}
		\captionof{figure}{Speedup vs. approximation,1m/7w cluster}
		\label{plt:speedupvsapprox}
	\end{minipage}
\end{center}
\begin{center}
	\begin{minipage}{0.49\linewidth}
		\includegraphics[width=\textwidth]{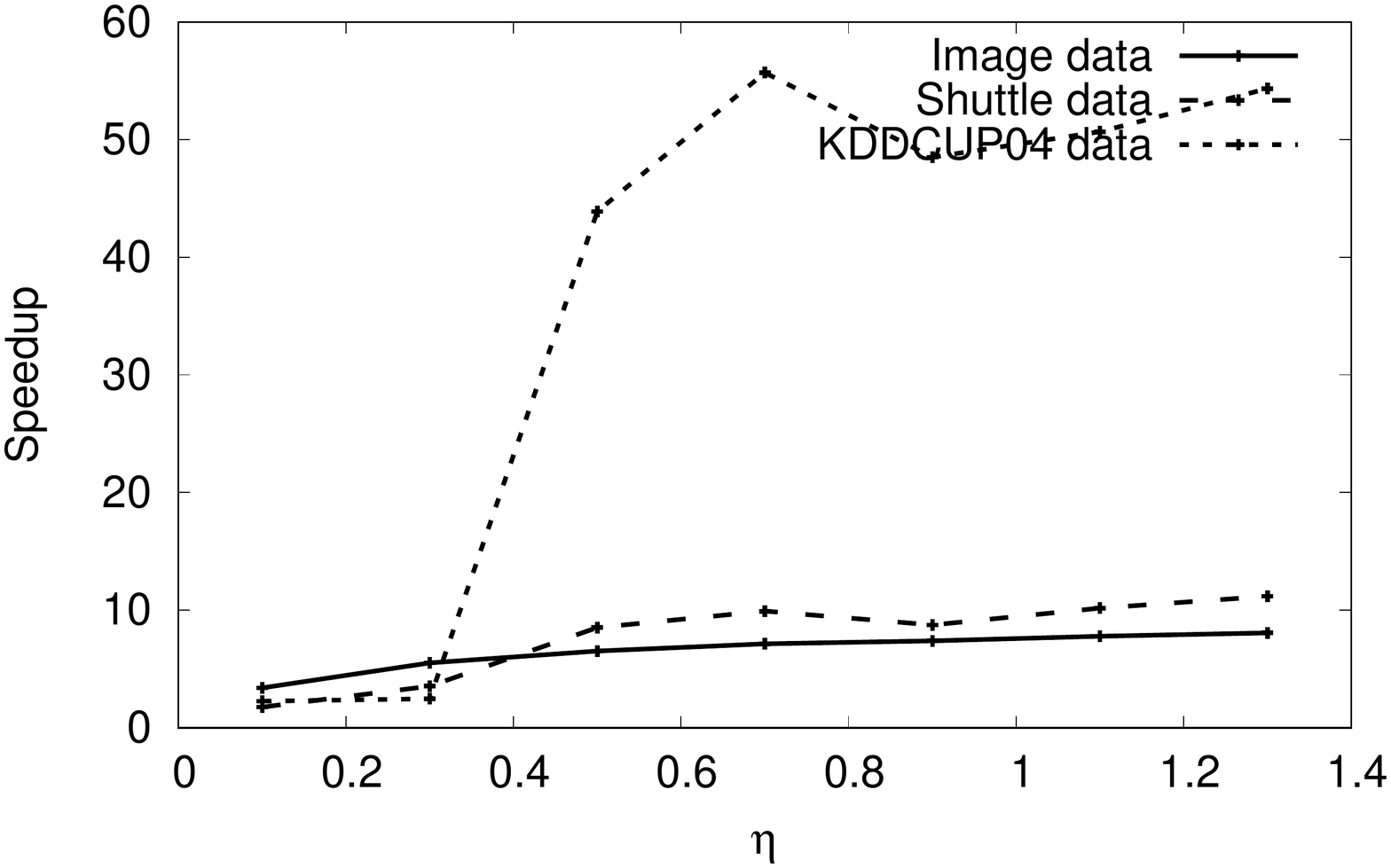}
		\captionof{figure}{Speedup vs $\eta$ parameter, 1m/3w cluster}
		\label{plt:speedupvseta2}
	\end{minipage}
	\begin{minipage}{0.49\linewidth}
		\includegraphics[width=\textwidth]{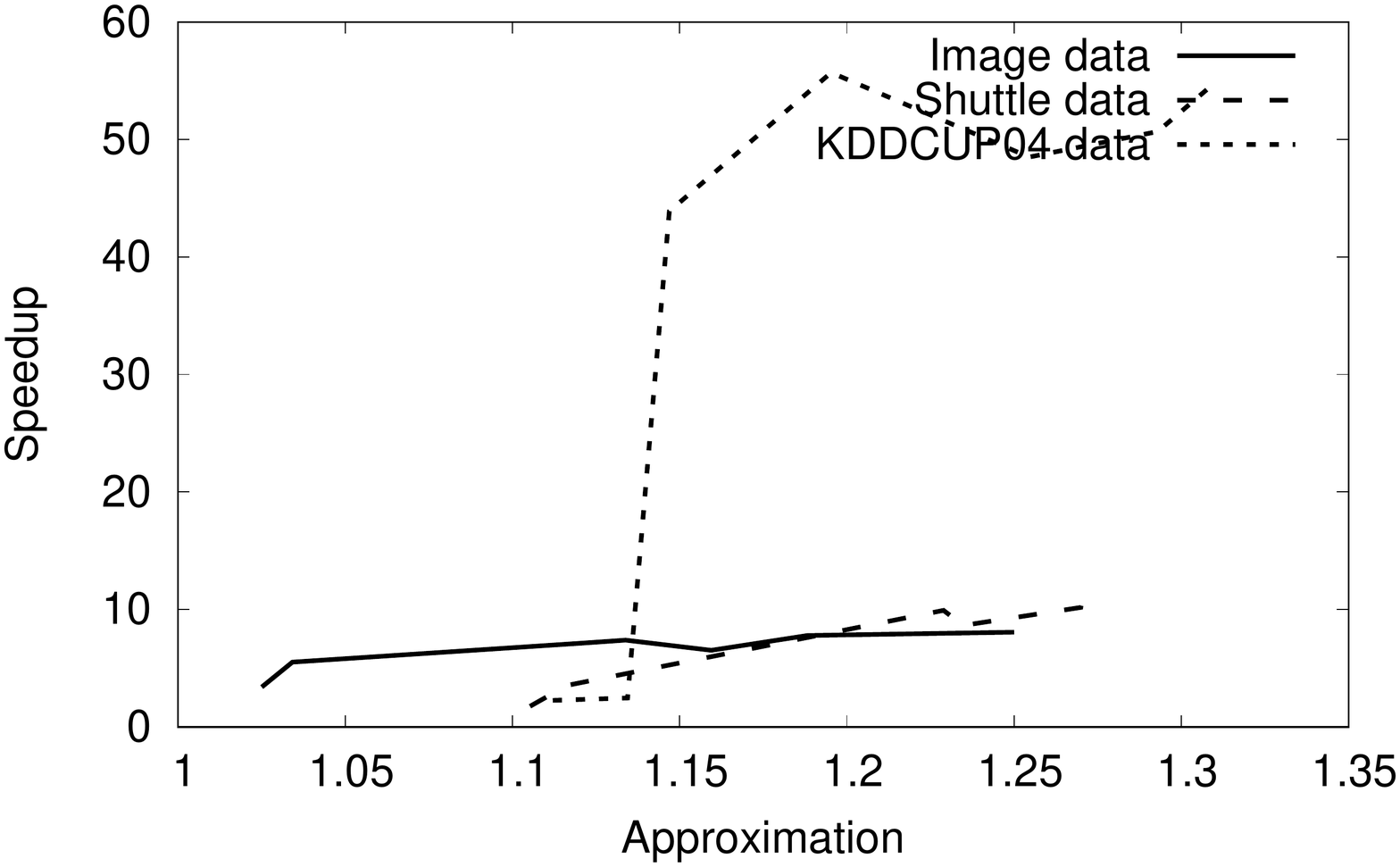}
		\captionof{figure}{Speedup vs. approximation,1m/3w cluster}
		\label{plt:speedupvsapprox2}
	\end{minipage}
\end{center}

\begin{center}
	\begin{minipage}{0.49\linewidth}
		\includegraphics[width=\textwidth, height=0.25\textheight, clip, trim=60pt 90pt 25pt 70pt]{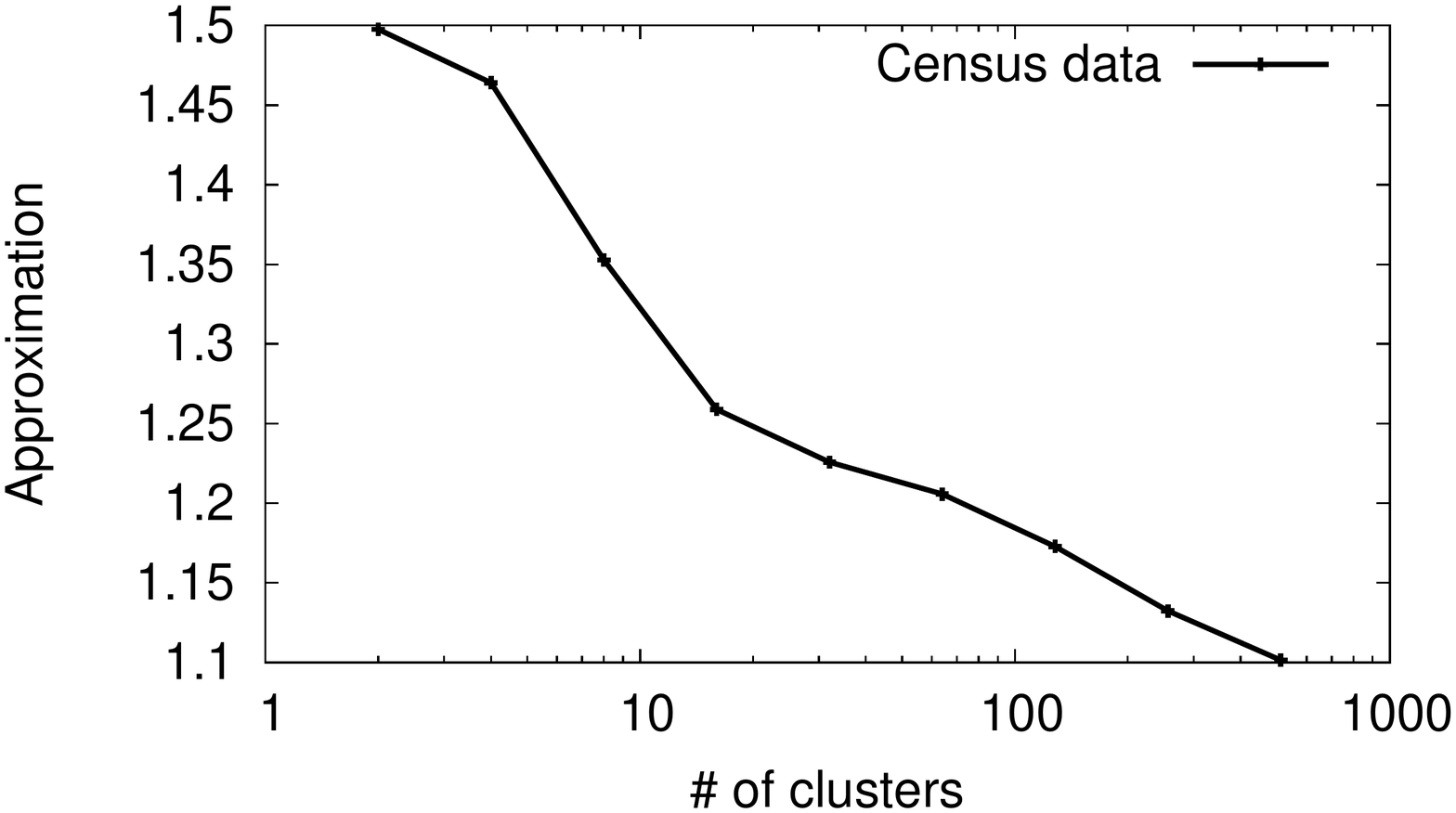}
		\captionof{figure}{Approximation vs. number of clusters}
		\label{plt:census-approxvk}
	\end{minipage}
	\begin{minipage}{0.49\linewidth}
		\includegraphics[width=\textwidth, clip, trim=60pt 60pt 25pt 40pt]{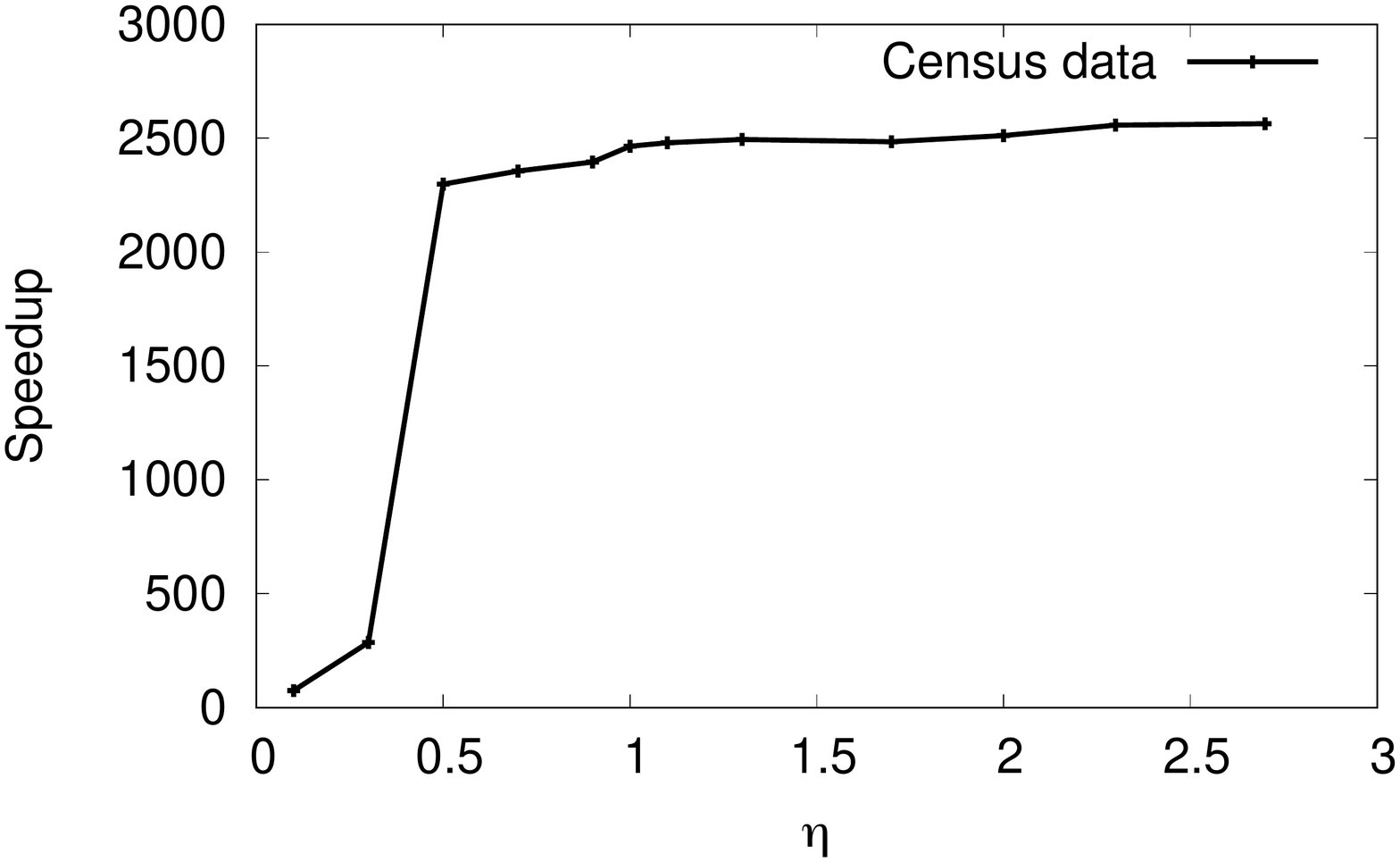}
		\captionof{figure}{Speedup vs. $\eta$}
		\label{plt:census-speedupveta}
	\end{minipage}
\end{center}

\begin{center}
	\begin{minipage}{0.49\linewidth}
		\includegraphics[width=\textwidth, clip, trim=60pt 60pt 25pt 40pt]{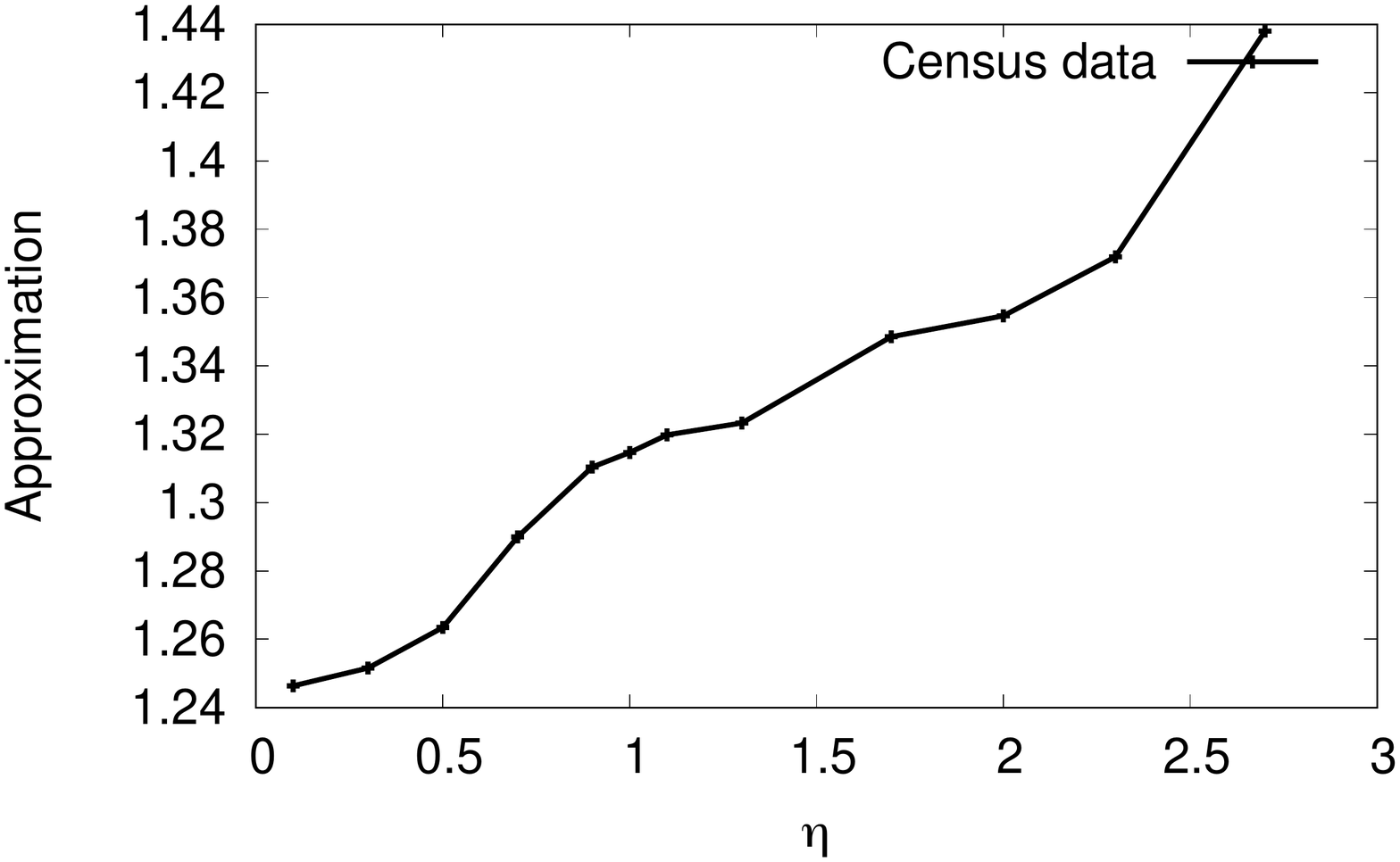}
		\captionof{figure}{Approximation vs $\eta$ parameter}
		\label{plt:census-approxveta}
	\end{minipage}
	\begin{minipage}{0.49\linewidth}
		\includegraphics[width=\textwidth, clip, trim=60pt 60pt 25pt 40pt]{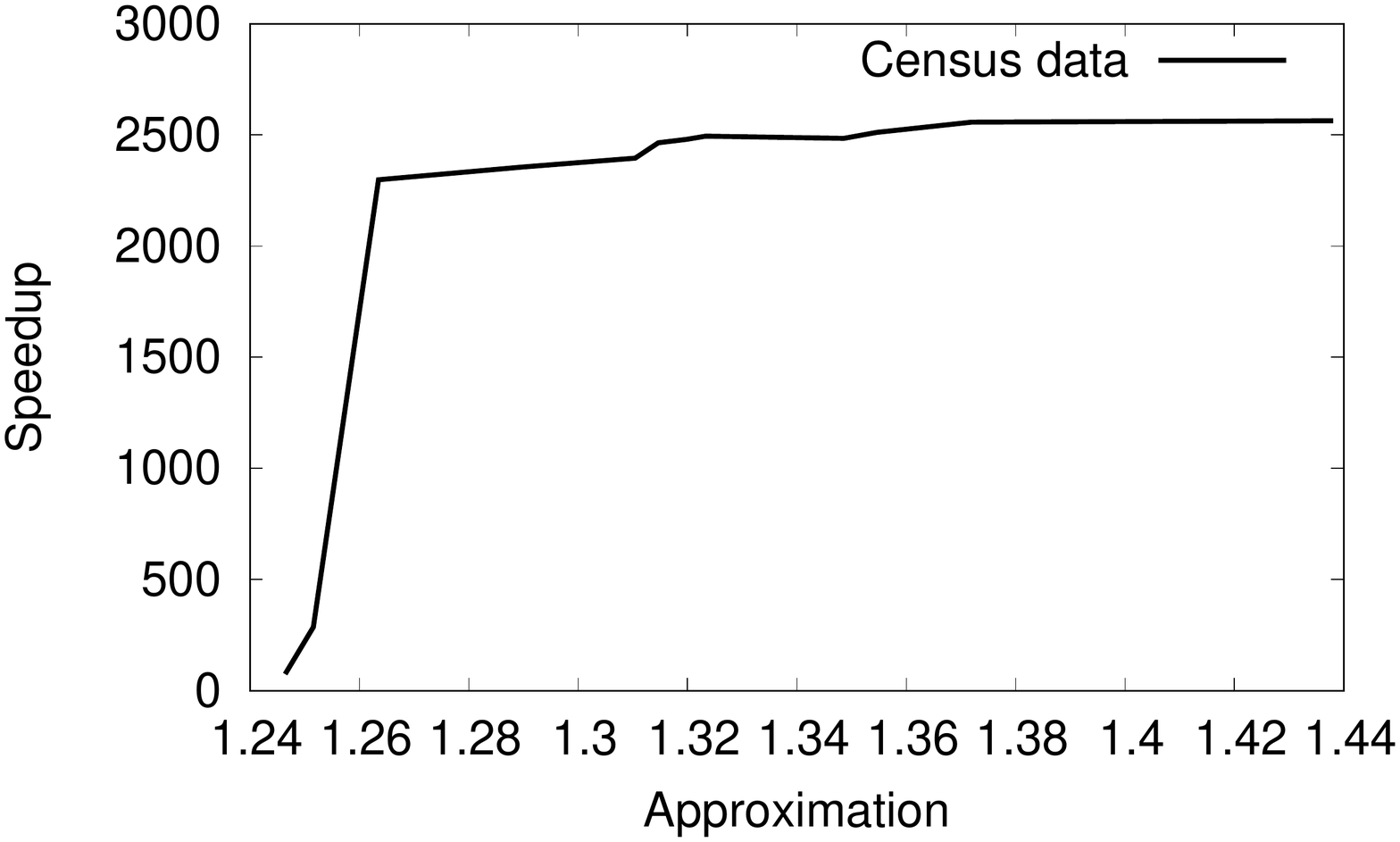}
		\captionof{figure}{Speedup vs. approximation}
		\label{plt:census-speedupvapprox}
	\end{minipage}
\end{center}

\subsubsection*{Results}

Figure~\ref{plt:approxvsk} shows dependence of apprxoimation to the $k$-SLC objective as a function of $k$ for $\eta = 0.5$.
Figure~\ref{plt:approxvseta} shows how approximation varies empirically as a function of $\eta$ (for $k = 10$).
Evaluation of time performance demonstrates more than an order of magnitude speedup over sequential Prim's algorithm\footnote{We use Prim's algorithm as a sequential benchmark as our local computation steps are also performed using Prim's algorithm in the experiments. For datasets of our size Prim's algorithm outperformed approximate nearest-neighbor (ANN) based algorithms locally and it was our goal to use the best local algorithm. In order to observe performance improvement from using ANN algorithms larger datasets are required due to several log-factors and a large constant in the theoretical almost-linear time complexity.} for $k=10$ as a function of $\eta$ parameter and empirical approximation. Results are given for two different cluster setups: 1m/7w ( Figure~\ref{plt:speedupvseta} and Figure~\ref{plt:speedupvsapprox}) and 1m/3w (Figure~\ref{plt:speedupvseta2} and Figure~\ref{plt:speedupvsapprox2}) and are averaged over multiple runs to ensure consistency.
We note that dramatic increase in speedup for the KDDCUP04 dataset around value $\eta = 0.5$ and approximation $1.15$ corresponds to the fact that local inputs start to fit in L2-cache which provides more than an order of magnitude improvement over RAM.

Figure~\ref{plt:census-approxvk} shows the dependence of approximation as a function of $k$, ($\eta = 1.3$) for the census data. Figure~\ref{plt:census-approxveta} shows how approximation varies as a function of $\eta$ (for $k = 10$) in the census dataset. Figure~\ref{plt:census-speedupveta} demonstrates that we achieve more than an order of magnitude of speed for the census data too. Figure~\ref{plt:census-speedupvapprox} shows the dependence of speedup as a function of approximation for the parameter $\eta = 1.3$. Once again  like on the KDDCUP04 dataset we observe a dramatic increase in the speedup at around $\eta = 0.5$ and approximation $1.15$ due to the fact the local inputs start to fit in L2-cache.

Figure~\ref{plt:census-speedupvapprox} shows the dependence of speedup as a function of approximation.  
We observe a dramatic increase in the speedup at around approximation $1.26$ due to the fact  local inputs start to fit in L2-cache.
Figure~\ref{plt:census-approxvk} shows the dependence of approximation as a function of $k$  for the census data.

In all cases our algorithms terminated in at most $40$ rounds. All experiments were performed under $\ell_2$-distance.  Since our $\ell_1$ and $\ell_\infty$ algorithms use the same partitioning scheme their performance under these metrics is similar.

\vspace{20pt}

\section{Conclusions}

We give $O(\log n)$-round MPC algorithms for $k$-Single-Linkage Clustering under popular distance metrics: $\ell_0, \ell_1,\ell_2$ and $\ell_\infty$. Our hardness results justify this round complexity assuming the most popular conjectures in the MPC literature. Experimental results show improvement in time of several orders of magntitude over sequential algorithms. 

Among interesting directions for further research we would like to highlight two: extending our results to higher dimensions and considering other popular linkage-based objectives (complete linkage and average linkage). In particular, we believe that some of our techniques should be applicable to the latter problem. Note that unlike for single-linkage there is no global guarantee about the outcome of complete-linkage and average-linkage clustering algorithms (see \cite{DL05}), hence one can only hope to achieve approximation for individual steps rather than for the final result. In particular, this justifies our focus on the single-linkage objective which is the only linkage-based objective for which sequential algorithms with global guarantees are known.

Another promising direction for future research is understanding the complexity of optimizing the entire hierarchy using a single objective as suggested by Dasgupta in~\cite{D16}.
Compared to the general graph metric setting which has been studied extensively~\cite{RP16,CC17} we are unaware of any such studies for vector data under $\ell_p$-distances.

\subsubsection*{Acknowledgments}
We would like to thank Alexandr Andoni, Aleksandar Nikolov and Krzysztof Onak who actively participated in the early stages of this project including multiple discussions of the MST algorithm of~\cite{ANOY14a} and its relationship to the Single-Linkage Clustering problem.
	
	\newpage
	
	\bibliography{slc-mrc}
	\bibliographystyle{abbrv}
	
	\newpage
	\appendix

\end{document}